\documentclass{article}
\usepackage{oldgerm,amsmath,pifont,amsfonts,amssymb,latexsym}
\usepackage{float,xspace,calc,theorem,ifthen}
\usepackage{enumerate,psboxit,subfigure}
\usepackage[dvips]{graphicx}
\usepackage{hyperref}
\usepackage{cite}
\usepackage{changebar}
\usepackage{verbatim}
\newtheorem{prop}{Proposition}

{\theorembodyfont{\rmfamily}\theoremstyle{plain}

}
{\theorembodyfont{\rmfamily}\theoremstyle{break}
}
\newcommand{\qed}{\phantom{xxxxxx}\hfill q.e.d.}
\newenvironment{proof}{{\noindent\em Proof:} }{\qed\\}

\newcommand{\rz}{{\mathbb R}}
\newcommand{\nz}{{\mathbb N}}

\newcommand{\zz}{{\mathbb Z}}

\newcommand{\quotemarks}[1]{`\emph{#1}'}
\begin{document}

\newcommand{\myfont}{\sf }

\title{\textbf{Linear and fractal diffusion coefficients in a family of one dimensional chaotic maps}}

\author{\textbf{Georgie Knight and Rainer Klages}\\
School of Mathematical Sciences, Queen Mary University of London,\\
Mile End Road, London E1 4NS, UK}

\maketitle


\begin{abstract}

We analyse deterministic diffusion in a simple, one-dimensional
setting consisting of a family of four parameter dependent, chaotic
maps defined over the real line. When iterated under these maps, a
probability density function spreads out and one can define a
diffusion coefficient.  We look at how the diffusion coefficient
varies across the family of maps and under parameter variation. Using
a technique by which Taylor-Green-Kubo formulae are evaluated in terms
of generalised Takagi functions, we derive exact, fully analytical
expressions for the diffusion coefficients. Typically, for simple maps
these quantities are fractal functions of control parameters. However,
our family of four maps exhibits both fractal and linear behavior. We
explain these different structures by looking at the topology of the
Markov partitions and the ergodic properties of the maps.\\

\end{abstract}

\section{Introduction}
\label{sec:intro}

Chaotic diffusion lies at the heart of the interaction between
transport theory and dynamical systems
\cite{Dorfman-book,Gasp-1998,Klages-07}. The study of this interaction
was started in the early 1980's
\cite{Geisel-1982,Schell-1982,Grossman-1982} and since then defines an active
area of research producing interesting results for both mathematicians
and physicists \cite{Cristadoro-2005,Keller-2008}. By analysing
transport theory in the context of dynamical systems, one can take
into account the deterministic equations of motion and therefore does
not need to rely on any stochastic based analysis.

In this paper, we evaluate the diffusion coefficients for a set of
piecewise linear, one-dimensional chaotic dynamical systems. There are
a few methods for analytically evaluating the diffusion coefficient of
a dynamical system; there is the cycle expansion technique described
in \cite{Cvit-chaos}, there is an interesting method based on the
zeros of the systems dynamical zeta function \cite{Cristadoro-2005},
there is the \quotemarks{first passage method} \cite{Klages-96} from
statistical physics which is based on the escape rate of the system
and involves computing the eigenmodes of the probability density and
there is the \quotemarks{twisted eigenstate method} described in
\cite{Klages-07}.  We are interested in how the diffusion coefficient
varies under parameter variation. The diffusion coefficient is often a
very complicated fractal function of the parameter in simple one
dimensional maps, this result first being reported in
\cite{Klages-95}. Unfortunately, analytically the above methods can
often only be worked out for specific values of control parameters, or
are otherwise rather difficult to handle. We will thus use a method
based on the Taylor-Green-Kubo formula
\cite{Dorfman-book,Gasp-1998,Klages-07} and certain fractal
\quotemarks{generalised Takagi functions} from \cite{Klages-96}, in
order to highlight this methods power at dealing with the complicated
structure that one typically finds in the diffusion
coefficient. Furthermore, to this end we choose our parameter in such
a way as to simplify the derivation of the diffusion coefficients
relative to the work in \cite{Klages-96}. Our main results are the
analytic expressions for the parameter dependent diffusion
coefficients of these maps, an explanation for the mixture of
linearity and fractality in the diffusion coefficients and the
discovery that the diffusion coefficients are actually very stable in
certain parameter ranges.

In section \ref{sec:family} we introduce the family of maps that are
under scrutiny before briefly describing the method we use to obtain
the parameter dependent diffusion coefficients in section
\ref{sec:diffcoeff}. In section \ref{sec:Eval_tak} we discuss and
define functional recursive relations in terms of infinite sums for
the interesting fractal Takagi functions that we meet. In section
\ref{sec:struc_diff} we analyse the structure of the diffusion
coefficients and explain the features that we observe, most
importantly we explain why we get a mixture of fractality and
linearity. We also employ a method based on the Markov partitions of
the maps to pinpoint exactly where the local extrema will be in these
fractal diffusion coefficients
\cite{Klages-95,Klages-99,Klages-96}. In addition we explore the fact
that in certain parameter regions, the diffusion coefficients are very
stable to dramatic changes in the microscopic dynamics of the maps,
this property being interestingly opposed to the fact that the
diffusion coefficients are extremely sensitive to parameter variation
in other areas of the parameter space. Section \ref{sec:conclusion} is
a conclusion.

\section{The family of maps}
\label{sec:family}

In this section the family of maps and how they are constructed will be introduced.

For $h \geq 0$ let $M_h(x):[0,1]\rightarrow \rz$ be a parameter dependent variant of the well known Bernoulli shift map. The parameter lifts the first branch, and lowers the second branch i.e.

\begin{equation}
M_h (x)=
\left\{
\begin{array}{rl}
2x + h & 0\leq x <\frac{1}{2}\\
2x -1 -h & \frac{1}{2}\leq x < 1\end{array}\right. .
\label{Eq:M_h_box}
\end{equation}

In order to create an extended system for diffusion, we define $M_h(x):\rz \rightarrow \rz$ by periodically copying equation (\ref{Eq:M_h_box}), with a lift of degree one such that

\begin{equation}
                M_h(x+n)=M_h(x)+n, \ \ n\in \zz.
\label{Eq:lift}
\end{equation}

The map described above and its diffusion coefficient were first studied in \cite{Gasp-Klages-98} using a different method. We call it the lifted Bernoulli shift map. This process of copying a map with a lift of degree one is a common way to create a diffusive map \cite{Geisel-1982,Grossman-1982,Schell-1982}. The use of the lift parameter $h$ ensures that the invariant probability density function (p.d.f) remains a constant function throughout the entire parameter range. This helps simplify the derivation of the diffusion coefficient as an invariant p.d.f is an essential ingredient in the Taylor-Green-Kubo formula that is used to derive the diffusion coefficient \cite{Dorfman-book,Klages-96}.

The remaining members of the family are created by changing the sign of the gradient in equation (\ref{Eq:M_h_box}). Let $W_h(x):[0,1]\rightarrow \rz$

\begin{equation}
W_h (x)=
\left\{
\begin{array}{rl}
-2x + h +1& 0\leq x <\frac{1}{2}\\
-2x +2 -h & \frac{1}{2}\leq x < 1\end{array}\right.  ,
\label{Eq:W_h_box}
\end{equation}

\noindent which we call the lifted negative Bernoulli shift map. Let $V_h(x):[0,1]\rightarrow \rz$

\begin{equation}
V_h (x)=
\left\{
\begin{array}{rl}
-2x +1+ h & 0\leq x <\frac{1}{2}\\
2x -1 -h & \frac{1}{2}\leq x < 1\end{array}\right.,
\label{Eq:V_h_box}
\end{equation}

\noindent which we call the lifted V map. Let $\Lambda_h(x):[0,1]\rightarrow \rz$

\begin{equation}
\Lambda_h (x)=
\left\{
\begin{array}{rl}
2x + h & 0\leq x <\frac{1}{2}\\
-2x +2 -h & \frac{1}{2}\leq x < 1\end{array}\right.  ,
\label{Eq:Lambda_h_box}
\end{equation}

\noindent which we call the lifted tent map. Again we apply the lift of degree one condition of equation (\ref{Eq:lift}) to equations (\ref{Eq:W_h_box}), (\ref{Eq:V_h_box}) and (\ref{Eq:Lambda_h_box}) to create spatially extended systems defined over the real line. See Figure \ref{fig:4maps} for an illustration.

\newpage

\begin{figure}[h!]
\begin{center}
  \includegraphics[width=6cm]{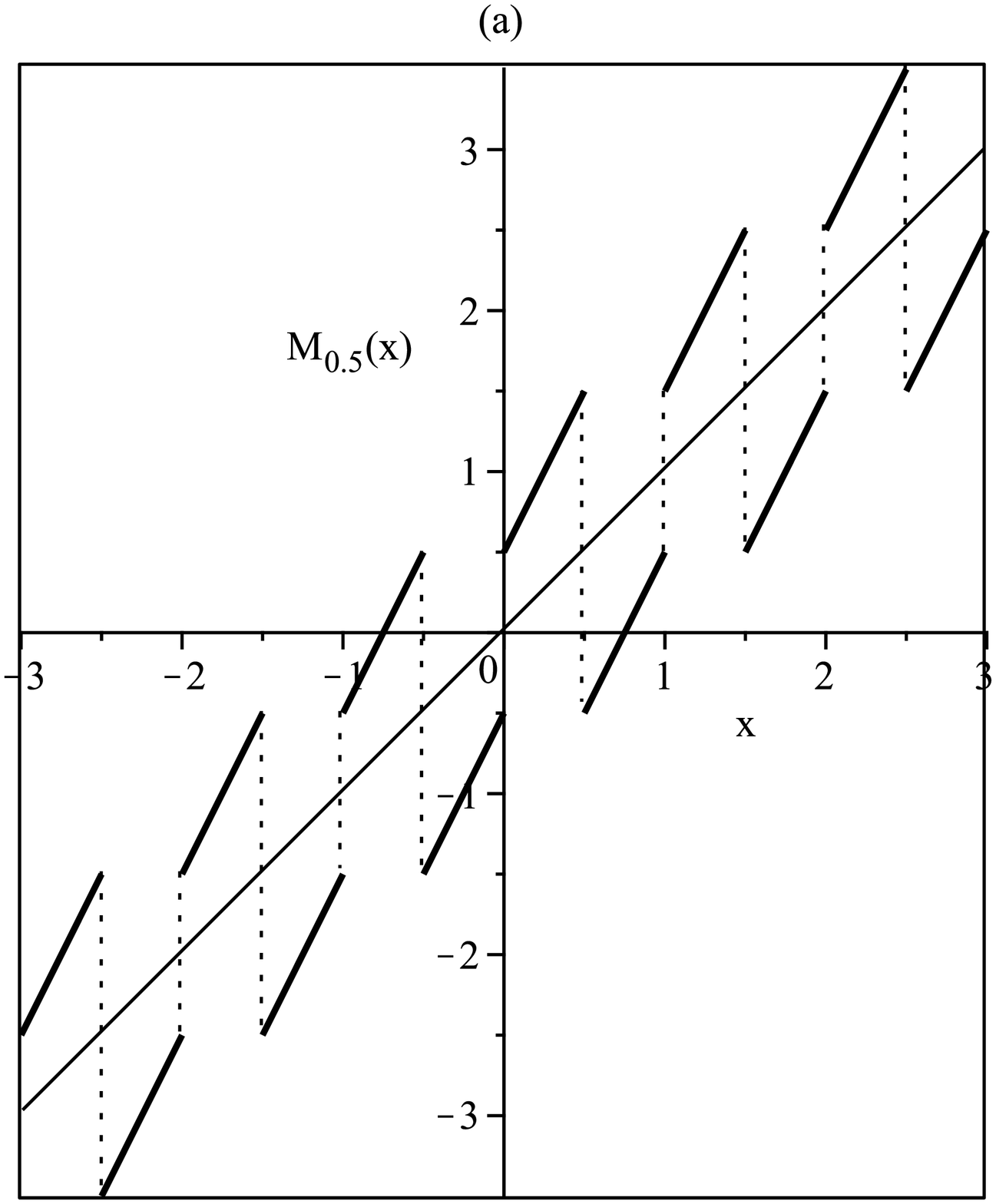} \includegraphics[width=6cm]{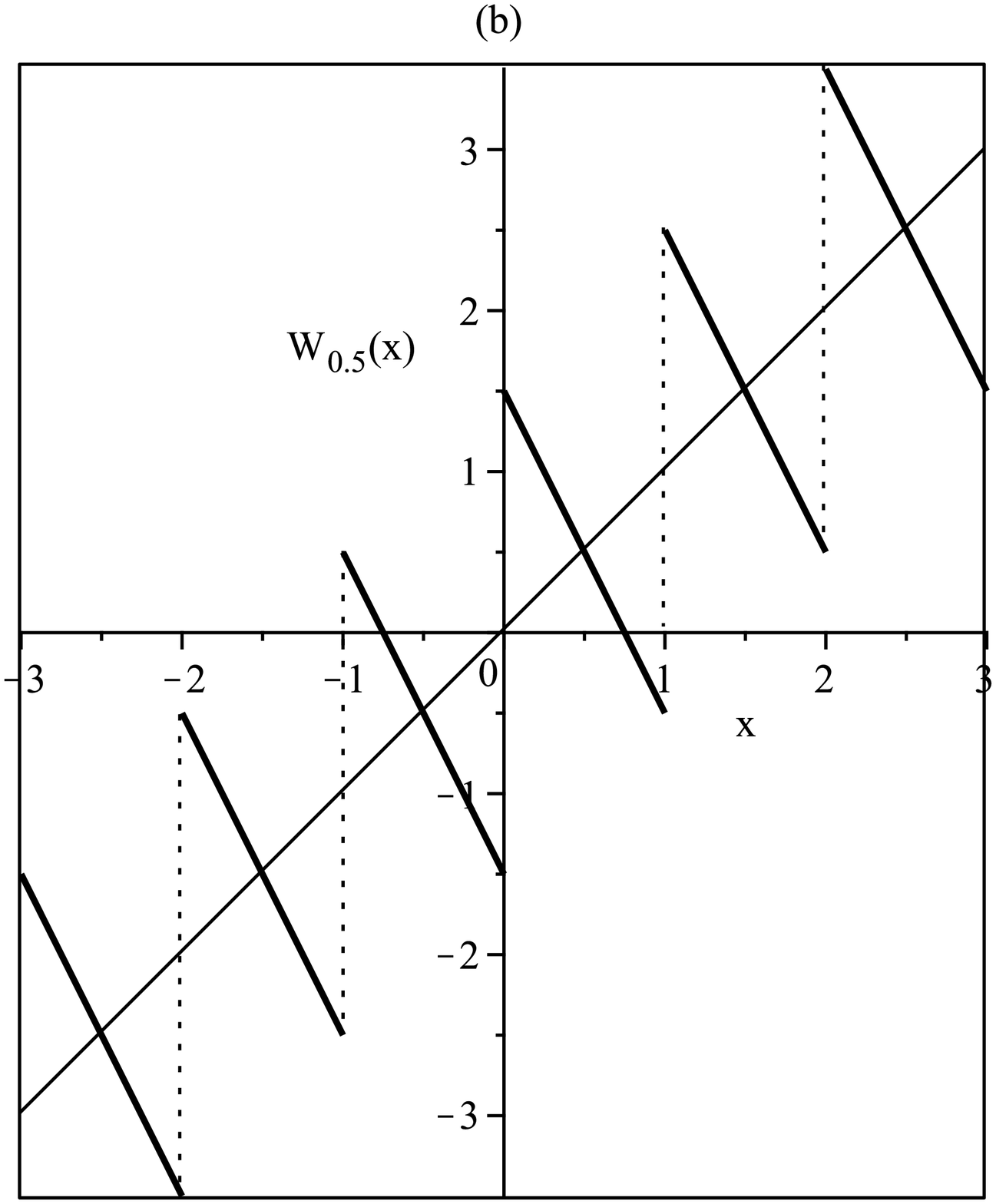}\\
\includegraphics[width=6cm]{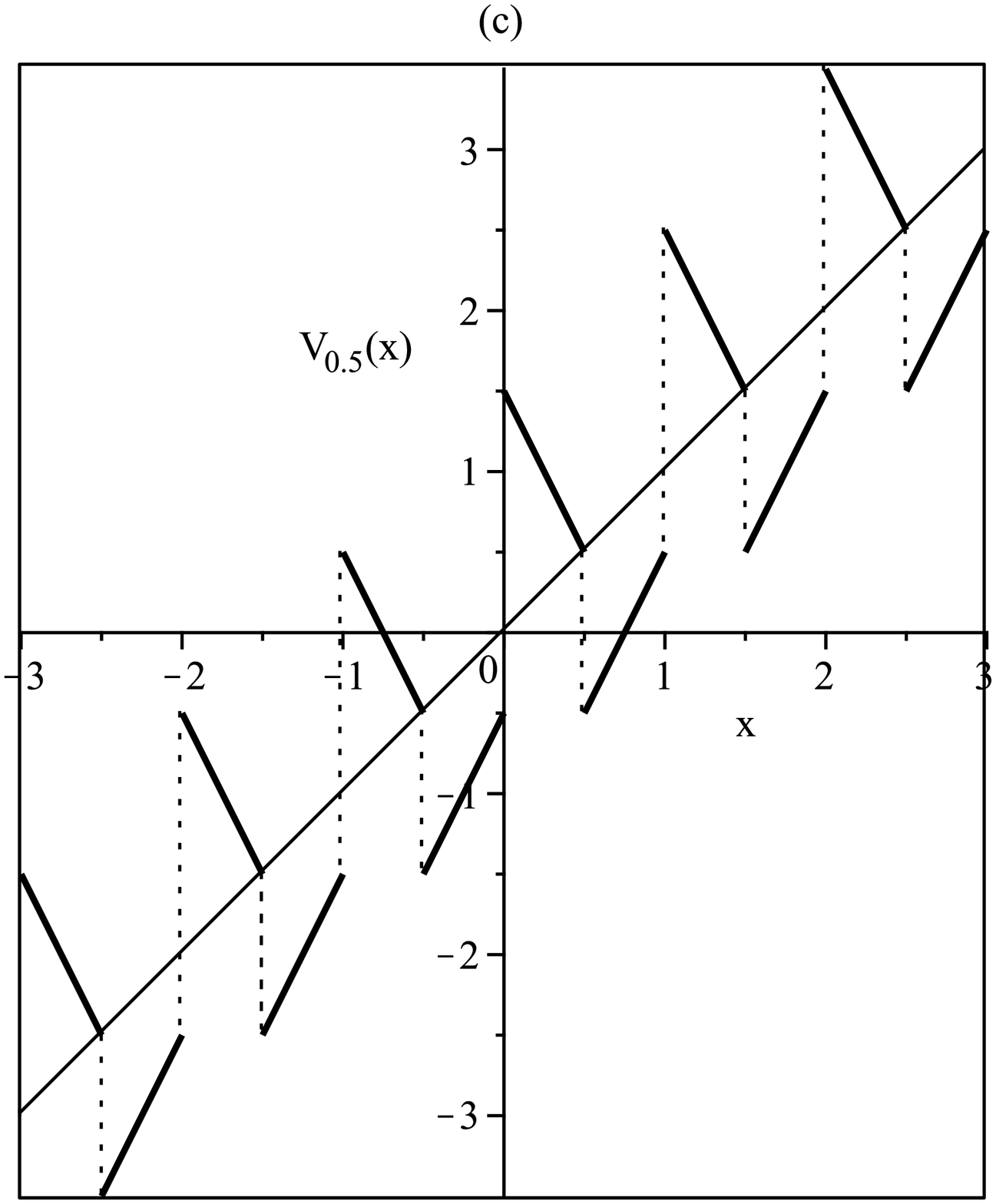} \includegraphics[width=6cm]{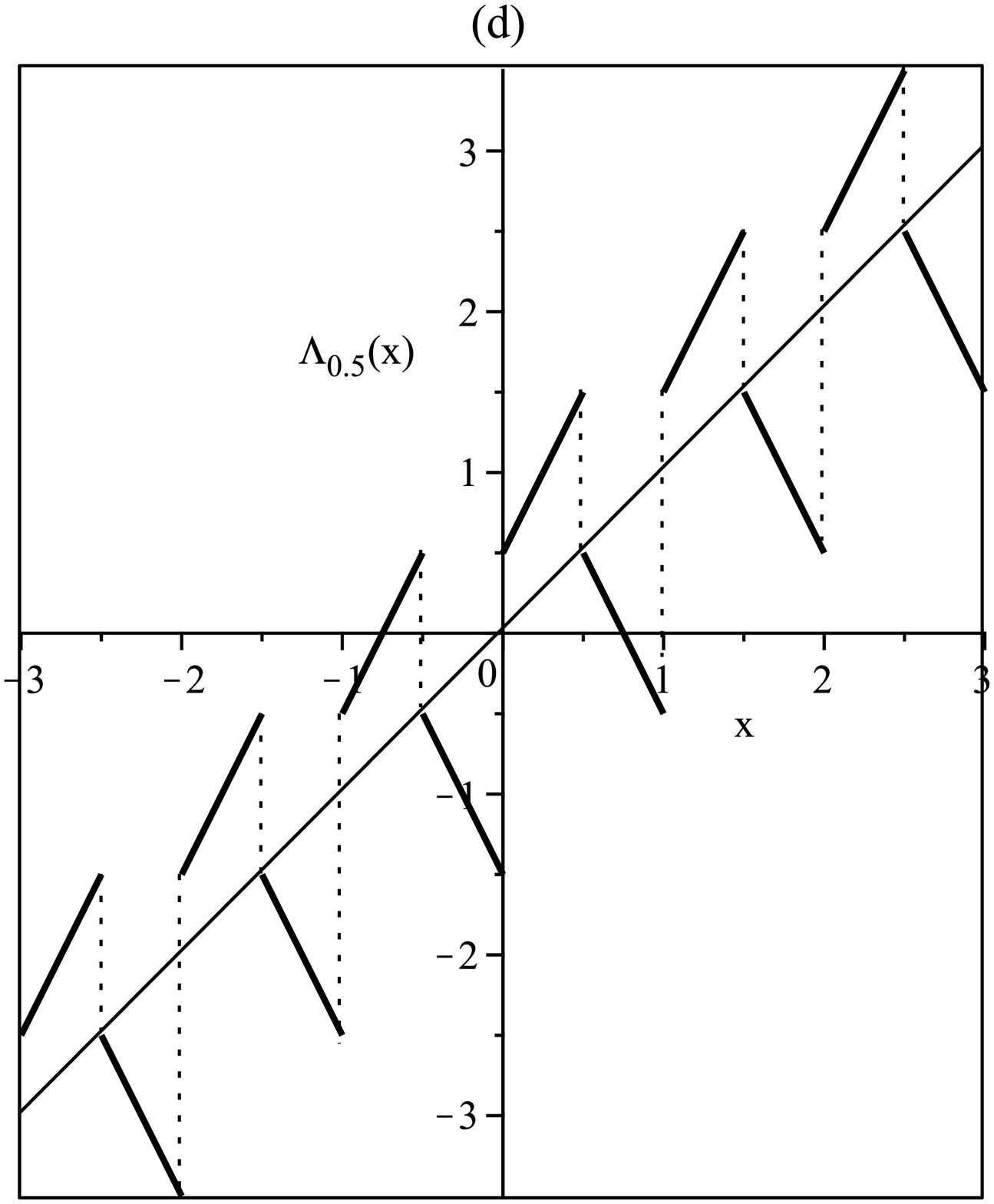}
\end{center}
\caption{\footnotesize{\emph{The family of maps}. In this figure, a section of each of the four maps under scrutiny is illustrated at a parameter value of $h=0.5$. In \textbf{(a)} the lifted Bernoulli shift, so-called because the Bernoulli shift is recovered when $h=0$ on the unit interval. In \textbf{(b)} the lifted negative Bernoulli shift, so-called because a version of the Bernoulli shift with a negative gradient is recovered when $h=0$. In \textbf{(c)} the lifted V map, so-called because a V map is recovered when $h=0$. In \textbf{(d)} the lifted tent map, so-called because a tent map is recovered when $h=0$.}}
\label{fig:4maps}
\end{figure}

\section{Deriving the diffusion coefficient}
\label{sec:diffcoeff}

The diffusion coefficient $D$ is given by Einstein's formula for diffusion in one dimension

\begin{equation}
D(h) = \lim_{n\to\infty}\frac{\left\langle (x_n-x_0)^{2}\right\rangle}{2n}
\label{Eq:Einstein_diff}
\end{equation}

\noindent where we are evaluating the diffusion coefficient as a function of the parameter $h$. The angular brackets $ \left\langle ... \right\rangle$ represent an average taken over the invariant p.d.f, $\rho^*(x)$

\begin{equation}
                \left\langle ... \right\rangle := \int_0^1 dx\rho^*(x)... \ \ .
\label{Eq:average}
\end{equation}

\noindent Equation (\ref{Eq:Einstein_diff}) can be transformed into the Taylor-Green-Kubo formula \cite{Dorfman-book}

\begin{equation}
D(h) = \lim_{n\to\infty} \left(\sum_{k=0}^n \left\langle \tilde{v}_0(x) \tilde{v}_k(x) \right\rangle\right) -\frac{1}{2}\left\langle \tilde{v}_0(x)^2 \right\rangle\
\label{Eq:TGK}
\end{equation}

\noindent where  $\tilde{v}_k(x):\: \rz \rightarrow \rz$,

\begin{equation}
               \tilde{v}_k(x):= x_{k+1} - x_k
\label{Eq:tildev}
\end{equation}

\noindent gives the displacement of a point $x$ at the $k^{th}$ iteration $x_k$. It can in fact be shown that we need only consider the integer value of the displacement of a point $x$  \cite{Klages-02}, i.e. we can replace $\tilde{v}_k(x)$ with the simpler function $v_k(x):\: \rz \rightarrow \zz$,

\begin{equation}
              v_k(x):=\lfloor x_{k+1} \rfloor-\lfloor x_k \rfloor\quad ,
\label{Eq:v_k}
\end{equation}

\noindent which for the lifted Bernoulli shift map $M_h(x)$ is

\begin{equation}
v_k (x)=
\left\{
\begin{array}{rl}
\left\lfloor h\right\rfloor & \ \ 0\leq \tilde{x}_k <\frac{1-\hat{h}}{2}\\
\left\lceil h\right\rceil & \ \ \frac{1-\hat{h}}{2} \leq \tilde{x}_k < \frac{1}{2}\\
-\left\lceil h\right\rceil & \ \  \frac{1}{2} \leq \tilde{x}_k < \frac{1+\hat{h}}{2}\\
-\left\lfloor h\right\rfloor & \ \ \frac{1+\hat{h}}{2}\leq \tilde{x}_k < 1\end{array}\right.
\label{Eq:v_h(x)bern}\quad ,
\end{equation}

\noindent where $\lfloor...\rfloor$ and $\lceil...\rceil$ are the floor function and ceiling function respectively and $\tilde{x}$ is $x$ modulo $1$. The function $\hat{h}: \rz\rightarrow [0,1]$ is defined as

\begin{equation}
\hat{h}=
\left\{
\begin{array}{rl}
1                & \ \ h \in \{\zz\}\\
h \mod{1}        & \ \ otherwise    \end{array}\right. \ \ .
\label{Eq:hath}
\end{equation}

 \noindent Equation (\ref{Eq:hath}) is simply a corrective function that ensures equation (\ref{Eq:v_h(x)bern}) is correct at the points of discontinuity. We can also use the fact that $\rho^*(x)=1$ in our model to simplify equation (\ref{Eq:TGK}) further to

\begin{equation}
D(h) = \lim_{n\to\infty} \left( \int_0^1 v_0(x) \sum_{k=0}^n v_k(x) \ \ dx \right) -\frac{1}{2}\int_0^1 v_0(x)^2 dx\ \ .
\label{Eq:TGK2}
\end{equation}

 \noindent Evaluating equation (\ref{Eq:TGK2}) will give us our diffusion coefficient. The second integral is simple enough, however the first integral is more involved. We will now explain the method we use for this.

 We first define a cumulative \quotemarks{jump function} $J_M^n(x):[0,1]\rightarrow \rz$ as

\begin{equation}
                J_M^n(x):= \sum_{k=0}^n v_k(x)
\label{Eq:J^n}
\end{equation}

\noindent which gives the integer displacement of a point $x$ after $n$ iterations. The subscript $M$ tells us we are considering the jump function for the lifted Bernoulli shift map $M_h(x)$. We substitute equation (\ref{Eq:J^n}) into equation (\ref{Eq:TGK2}) to obtain

\begin{equation}
                D_M(h) = \lim_{n \rightarrow \infty} \left(\int_0 ^1 v_0(x) J_M ^n (x)\ \  dx \right) -\frac{1}{2} \int_0^1 v_0^2(x)\ \  dx.
\label{Eq:Taylor-Green-Kubo}
\end{equation}

\noindent In order to extract the information we need from (\ref{Eq:Taylor-Green-Kubo}), we define $T_M(x):[0,1]\rightarrow \rz$ which integrates over the jump function, as in the first integral of (\ref{Eq:Taylor-Green-Kubo}),

\begin{equation}
               T_M(x):= \lim_{n \rightarrow \infty}T_M^n(x)=\lim_{n \rightarrow \infty}\int_0^x J_M^n(y)dy.
\label{Eq:Takagi_def}
\end{equation}

\noindent Equation (\ref{Eq:Takagi_def}) defines the \quotemarks{generalised Takagi functions} discussed above and in \cite{Klages-96}. Due to the chaotic nature of the maps, and in particular, the sensitive dependence on initial conditions, the jump function behaves very erratically for large $n$. This is reflected in $T_M(x)$ which becomes fractal in the limit as $n\rightarrow\infty$. We call these functions \quotemarks{generalised Takagi functions} because for the lifted Bernoulli shift map with parameter $h=1$, one obtains the function first studied by T.Takagi in 1903 \cite{Takagi-1903}. This function is also an example of a \quotemarks{De Rham} function \cite{DeRham-1957}. T.Takagi was interested in this function from an analytical perspective as a function that is both continuous and non-differentiable.

In order to control the chaotic nature of the jump function, we derive a functional recursive relation \cite{Klages-96}. For the lifted Bernoulli shift map

\begin{equation}
J_M ^n (x)      = v_0(x) +J_M ^{n-1}\left(\tilde{M}_h (x)\right) \ \
\label{Eq:J_hrecursive}
\end{equation}

\noindent where $\tilde{M}_h (x)$ is equation (\ref{Eq:M_h_box}) taken modulo $1$. In turn, we define a functional recursive relation for the generalised Takagi functions, by substituting equation (\ref{Eq:J_hrecursive}) into equation (\ref{Eq:Takagi_def})

\begin{eqnarray}\nonumber
                T_M(x)     &=&     \lim_{n\rightarrow\infty} T_M^n(x)\\
                           &=&     \lim_{n\rightarrow\infty} \left( t_M(x) + \frac{1}{2}T_M^{n-1}\left(\tilde{M}_h(x)\right) \right),
\label{Eq:T(x)_recursion1}
\end{eqnarray}

\noindent where $t_M(x):[0,1]\rightarrow \rz$ is defined as

\begin{eqnarray}\nonumber
                         t_M(x)   &:=& \int_0^x v_0(y)dy\\
                                  &=& xv_0(x) +c.
\label{Eq:t(x)}
\end{eqnarray}

\noindent The constant of integration $c$ is evaluated by requiring the continuity of the Takagi function and that $T(0)=T(1)=0$. We also use the symmetry of the Takagi functions to simplify the evaluation. For the lifted Bernoulli shift map we obtain the recursion relation

\begin{equation}
T_M (x)=
\left\{
\begin{array}{lc}
\frac{1}{2}T_M \left(2x + \hat{h}\right) +\left\lfloor h\right\rfloor  x -\frac{1}{2}T_M (\hat{h})& \ \ 0\leq x <\frac{1-\hat{h}}{2}\\
\frac{1}{2}T_M \left(2x +\hat{h} -1\right) + \left\lceil h\right\rceil  x + \frac{\hat{h}-1}{2} - \frac{1}{2}T_M (\hat{h})& \ \ \frac{1-\hat{h}}{2} \leq x < \frac{1}{2}\\
\frac{1}{2}T_M \left(2x -\hat{h}\right) + \frac{1+\hat{h}}{2} - \left\lceil h\right\rceil  x - \frac{1}{2}T_M (\hat{h})& \ \ \frac{1}{2} \leq x < \frac{1+\hat{h}}{2}\\
\frac{1}{2}T_M \left(2x -1 -\hat{h}\right) -\left\lfloor h\right\rfloor  x + \left\lfloor h\right\rfloor - \frac{1}{2}T_M (\hat{h})& \ \ \frac{1+\hat{h}}{2}\leq x \leq 1\end{array}\right. \:,
\label{Eq:T_hfull}
\end{equation}

\noindent where we have taken the limit $n\rightarrow\infty$. We now have the ingredients that we need to derive the parameter dependent diffusion coefficient. We firstly apply equation (\ref{Eq:v_h(x)bern}) to the Taylor-Green-Kubo formula to obtain

\begin{eqnarray}\nonumber
D_M(h) &=& \lim_{n \rightarrow \infty} \left(\int_{0} ^{1} v_0 (x) J_M ^n (x)\ \  dx\right)- \frac{1}{2} \int_{0} ^{1} v_0 ^2 (x)\ \  dx\\
     &=& \lim_{n \rightarrow \infty}  \int_{0} ^{\frac{1-\hat{h}}{2}}\left\lfloor h\right\rfloor  J_M ^n (x)\ \  dx +
         \int_{\frac{1-\hat{h}}{2}} ^{\frac{1}{2}}\left\lceil h\right\rceil  J_M ^n (x)\ \  dx -
         \int_{\frac{1}{2}} ^{\frac{1+\hat{h}}{2}}\left\lceil h\right\rceil  J_M ^n (x)\ \  dx\\\nonumber
     &-& \int_{\frac{1+\hat{h}}{2}} ^{1} \left\lfloor h\right\rfloor  J_M ^n (x)\ \  dx   \ \
          - \ \ \frac{1}{2}\left( \int_{0}^{\frac{1-\hat{h}}{2}}\left\lfloor h\right\rfloor ^2+
         \int_{\frac{1-\hat{h}}{2}} ^{\frac{1}{2}}\left\lceil h\right\rceil ^2+
         \int_{\frac{1}{2}} ^{\frac{1+\hat{h}}{2}}\left\lceil h\right\rceil ^2+
         \int_{\frac{1+\hat{h}}{2}} ^{1} \left\lfloor h\right\rfloor ^2 \right).
\label{Eq:D(h)_der_1}
\end{eqnarray}

\noindent Evaluating the integrals, simplifying using equation (\ref{Eq:T_hfull}) and gathering relevant terms we obtain

\begin{eqnarray}
D_M(h) &=& \left(\left\lfloor h\right\rfloor - \left\lceil h\right\rceil\right)\left(\left\lfloor h\right\rfloor-
         \hat{h}  \left\lfloor h\right\rfloor -T_M\left(\hat{h}\right)\right) +\left\lceil h\right\rceil \left( \left\lceil h\right\rceil + \hat{h} -1             \right)\\\nonumber
     &-& \frac{1}{2}\left(\left\lfloor h\right\rfloor ^2\left(1-\hat{h}\right)+\left\lceil h\right\rceil ^2 \hat{h}\right),\nonumber
\label{Eq:D(h)_der_2}
\end{eqnarray}

\noindent which after some wrangling can be rewritten as

\begin{equation}
D_M(h)     = \frac{\left\lceil h\right\rceil ^2}{2}+\left(\frac{1-\hat{h}}{2}\right)\left(1-2  \left\lceil h\right\rceil \right)+T_M\left(\hat{h}\right) \ \ .
\label{Eq:D(h)_der_3}
\end{equation}

\noindent Equation (\ref{Eq:D(h)_der_3}) is an analytic expression for the parameter dependent diffusion coefficient for the lifted Bernoulli shift map. The first two terms in equation (\ref{Eq:D(h)_der_3}) form a piecewise linear function that is asymptotically equal to $\frac{h^2}{2}$ for large $h$. This term tells us how the diffusion coefficient grows for large $h$. Interestingly, our shifted map belongs to a different universality class compared to the maps studied in \cite{Klages-96,Klages-97} where the gradient was varied as a parameter. There it was conjectured that the parameter dependent diffusion coefficient for this class of maps would grow quadratically with the height parameter $h$ with a coefficient of $\frac{1}{6}$.

The $T_M\left(\hat{h}\right)$ term tells us about the fine scale
structure. This fine scale structure is periodic modulo $1$ as it is a
function of $\hat{h}$, see figure \ref{Fig:D(h)_bern_periodicity} for
an illustration. We also observe a second region of asymptotic
behaviour where the diffusion coefficient tends to $\frac{h}{2}$ as $h
\rightarrow 0$, i.e.\ it behaves linearly and reproduces a simple
random walk result for diffusion. This result agrees with the findings
in \cite{Klages-96,Klages-97}, it tells us that the higher order
correlations of our system are negligible as $h \rightarrow 0$. This
change from one type of asymptotic behaviour to another in a system
was denoted as a crossover in deterministic diffusion
\cite{Klages-97}.

Although the calculations have been presented for the lifted Bernoulli
shift, the method is the same for the other three maps. If we turn our
attention to the lifted negative Bernoulli shift map $W_h(x)$ the
Takagi function is given by

\begin{equation}
T_W (x)=
\left\{
\begin{array}{lc}
-\frac{1}{2}T_W \left(-2x + \hat{h}\right) +\left\lceil h\right\rceil  x+ \frac{1}{2}T_W (\hat{h})& \ \ 0\leq x <\frac{\hat{h}}{2}\\
-\frac{1}{2}T_W \left(-2x +\hat{h} +1\right) +\left\lfloor h\right\rfloor  x+ \frac{\hat{h}}{2}\left( \lceil h \rceil - \lfloor h \rfloor\right) + \frac{1}{2}T_W (\hat{h})& \ \ \frac{\hat{h}}{2} \leq x < \frac{1}{2}\\
-\frac{1}{2}T_W \left(-2x +2-\hat{h}\right)- \left\lfloor h\right\rfloor x  +\lfloor h \rfloor+  \frac{\hat{h}}{2}\left( \lceil h \rceil - \lfloor h \rfloor\right)+\frac{1}{2}T_W (\hat{h})& \ \ \frac{1}{2} \leq x < 1-\frac{\hat{h}}{2}\\
-\frac{1}{2}T_W \left(-2x +3 -\hat{h}\right)-\left\lceil h\right\rceil  x  + \left\lceil h\right\rceil + \frac{1}{2}T_W (\hat{h})& \ \ 1-\frac{\hat{h}}{2}\leq x \leq 1\end{array}\right.
\label{Eq:T_hfull_negbern}
\end{equation}

\noindent and the corresponding expression for the parameter dependent diffusion coefficient is,

\begin{equation}
                D_W(h)= \frac{\left\lfloor h \right \rfloor ^2}{2} + \frac{\hat{h}}{2}\left( \left\lceil h \right\rceil ^2 -\lfloor h \rfloor ^2\right) +T_W(\hat{h}).
\label{Eq:D(h)_negbern}
\end{equation}

\noindent The first two parts of equation (\ref{Eq:D(h)_negbern}) form the same piecewise linear function found in equation (\ref{Eq:D(h)_der_3}) so as $h \rightarrow \infty$ we observe the same asymptotic behaviour found in the lifted Bernoulli shift map, see figure \ref{Fig:D(h)_bern_periodicity}.



\begin{figure}[h]
\includegraphics[width=7cm]{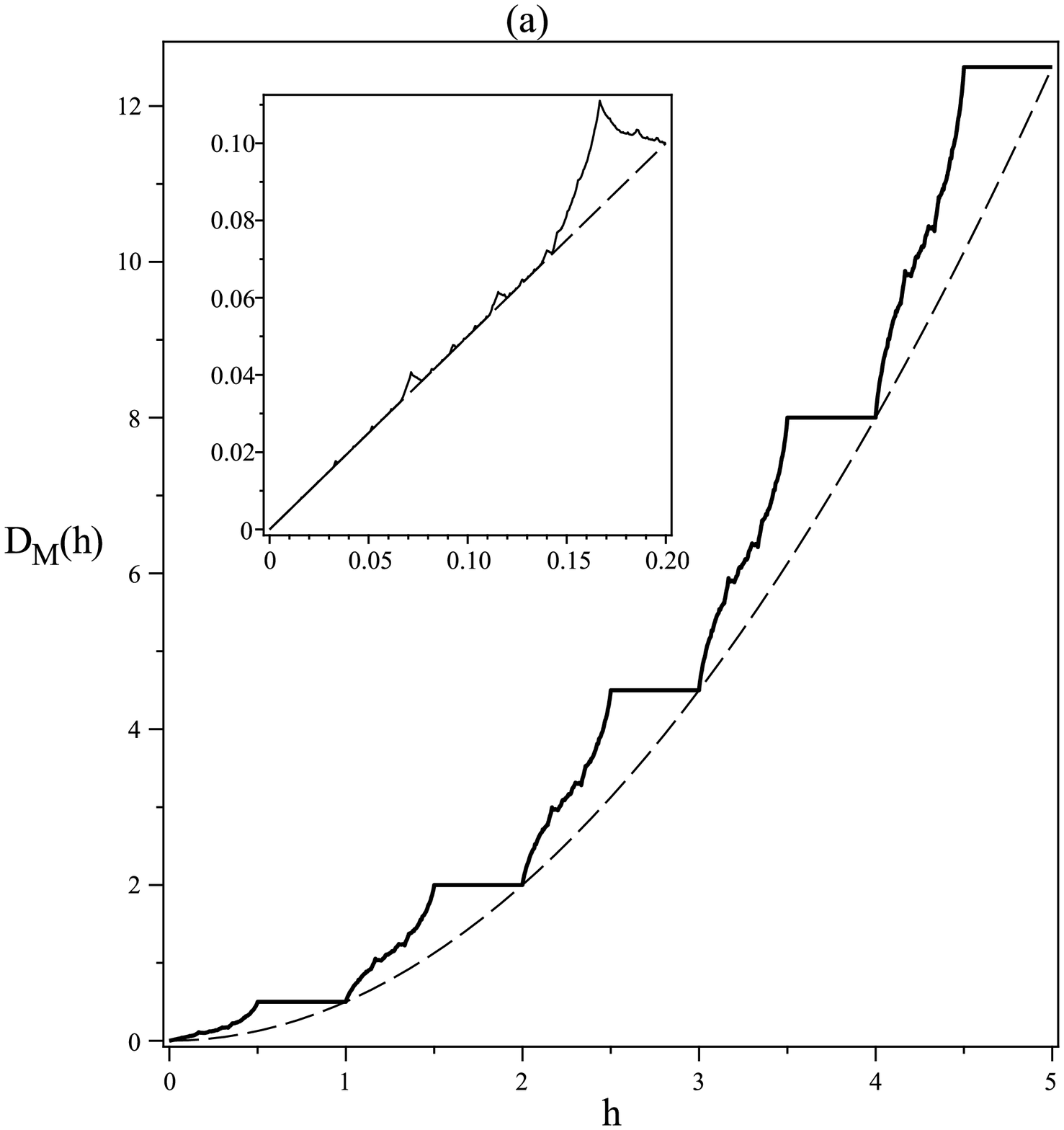}  \includegraphics[width=7cm]{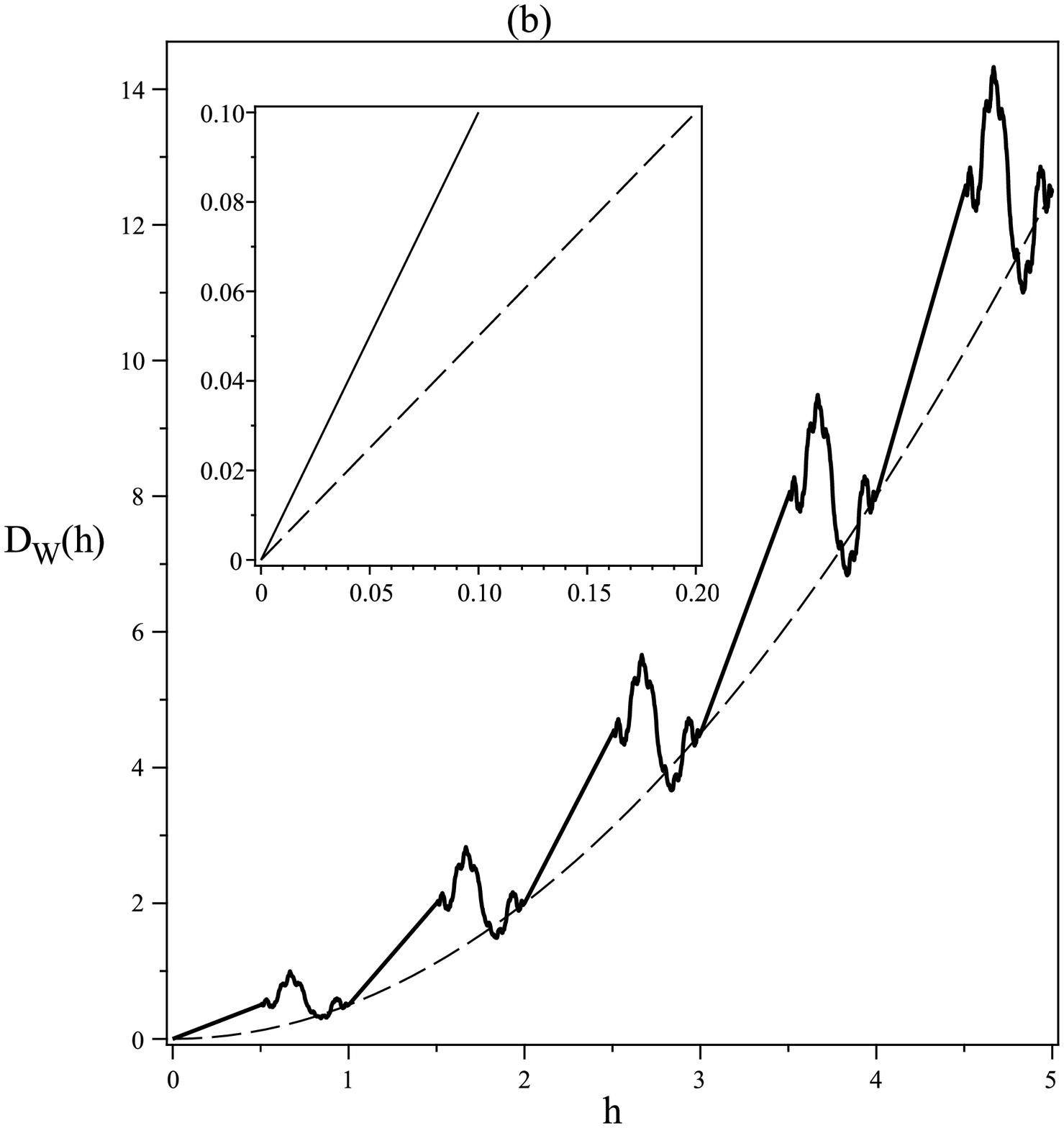}
\caption{\footnotesize{\emph{Large scale structure and asymptotic behaviour}. In this figure, the diffusion coefficients for the lifted Bernoulli shift (a) and the lifted negative Bernoulli shift (b) are illustrated. Also included in both is $f(h)=h^2/2$ to show how the function grows for large $h$. Note the periodicity of the fine scale structure. Inset in both is an illustration of the asymptotic behaviour as $h \rightarrow 0$.The dashed lines are $\frac{h}{2}$} to show the different behaviour in the two maps.}
\label{Fig:D(h)_bern_periodicity}
\end{figure}

The parameter dependent Takagi functions for the lifted tent map and lifted V map are different in character to the two Bernoulli shift maps discussed above. In order to emphasise this difference we restrict the parameter to $h \in [0,1]$. For the lifted tent map the Takagi function is

\begin{equation}
T_\Lambda (x)=
\left\{
\begin{array}{ll}
\frac{1}{2}T_\Lambda(2x+h)-\frac{1}{2}T_\Lambda(h)                           &\ \ 0 \leq x < \frac{1-h}{2}\\
x+\frac{1}{2}T_\Lambda(2x+h-1)+ \frac{h-1}{2}-\frac{1}{2}T_\Lambda(h)        &\ \ \frac{1-h}{2} \leq x < \frac{1}{2}\\
-\frac{1}{2}T_\Lambda(-2x+2-h)+\frac{h}{2}+\frac{1}{2}T_\Lambda(1-h)         &\ \ \frac{1}{2}\leq x < 1-\frac{h}{2}\\
-x-\frac{1}{2}T_\Lambda(-2x+3-h)+1+\frac{1}{2}T_\Lambda(1-h)                 &\ \ 1-\frac{h}{2} \leq x < 1 \end{array}\right. \ \ .
\label{Eq:T(h)_tent}
\end{equation}

\noindent Whilst the Takagi function for the lifted V map looks like

\begin{equation}
T_V (x)=
\left\{
\begin{array}{ll}
x-\frac{1}{2}T_V(-2x+h)+\frac{1}{2}T_V(h)                                      &\ \ 0 \leq x < \frac{h}{2}\\
-\frac{1}{2}T_V(-2x+1+h)+ \frac{h}{2}+\frac{1}{2}T_V(h)                        &\ \ \frac{h}{2} \leq x < \frac{1}{2}\\
-x+\frac{1}{2}T_V(2x-h)+\frac{h}{2}+\frac{1}{2}-\frac{1}{2}T_V(1-h)            &\ \ \frac{1}{2}\leq x < \frac{1+h}{2}\\
\frac{1}{2}T_V(2x-1-h)-\frac{1}{2}T_V(1-h)                                     &\ \ \frac{1+h}{2} \leq x < 1 \end{array}\right. \ \ .
\label{Eq:T(h)_V}
\end{equation}

\noindent The thing to notice here is the inclusion of the $T(1-h)$ terms. Usually we would use the symmetry of the Takagi functions to equate these terms with $T(h)$, but the Takagi functions for the lifted V map and lifted tent map do not have this symmetry. We will discuss the consequences of this and the diffusion coefficients for these two maps in sections \ref{sec:Eval_tak} and \ref{sec:struc_diff}.

\section{Evaluating the Takagi functions}
\label{sec:Eval_tak}

\begin{figure}[h!]
\begin{center}
  \includegraphics[width=7cm]{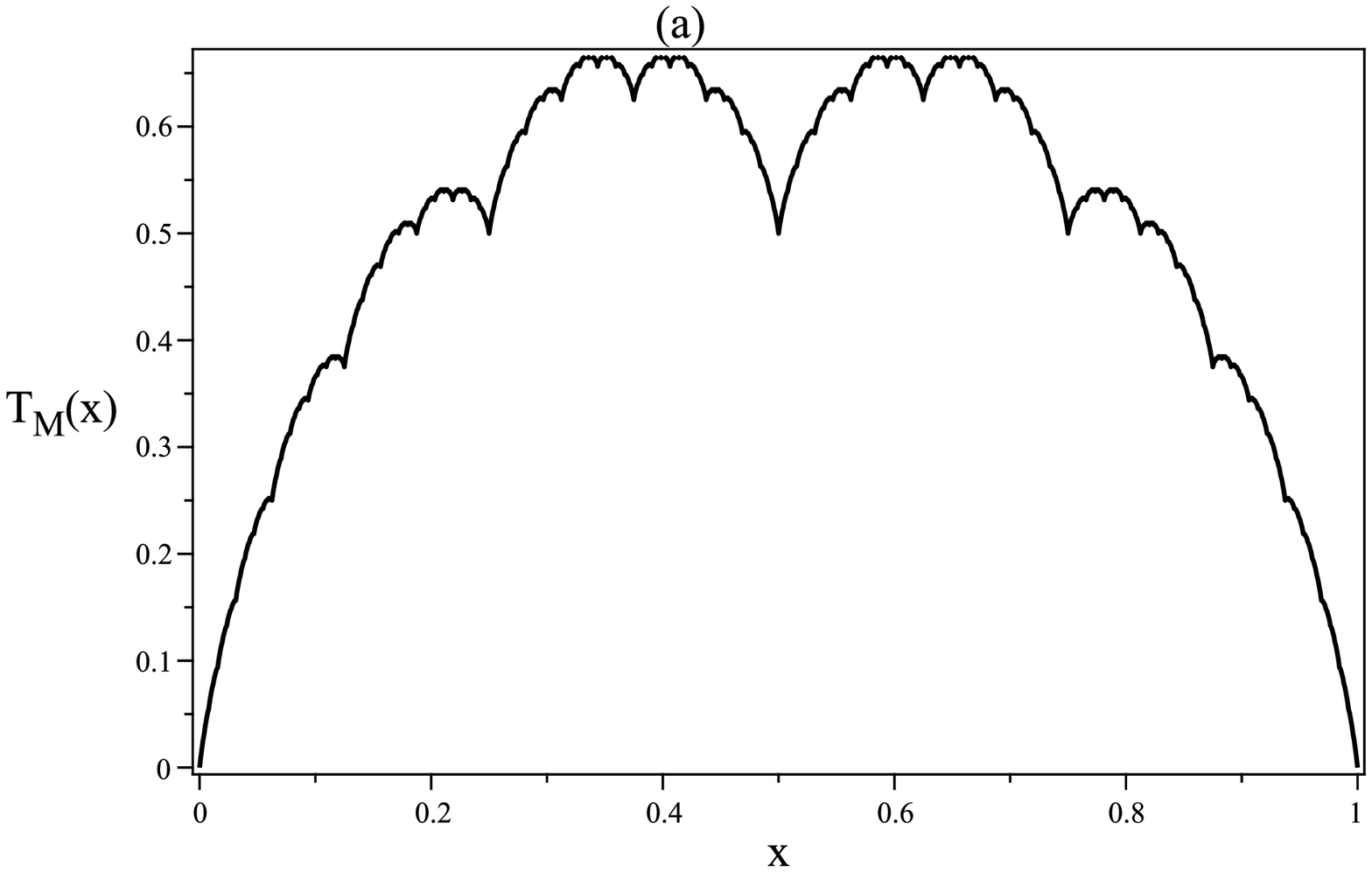} \includegraphics[width=7cm]{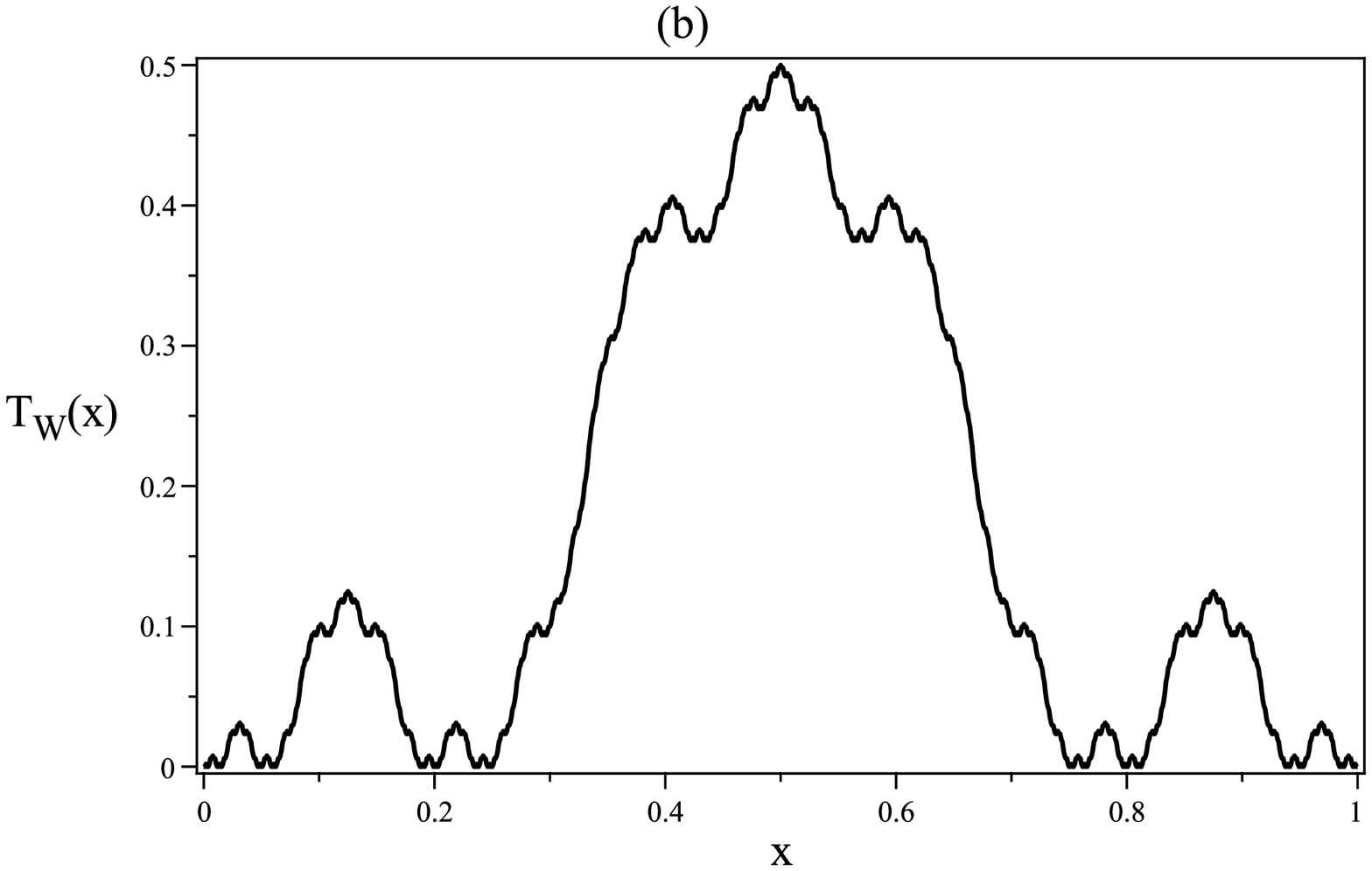}\\
\includegraphics[width=7cm]{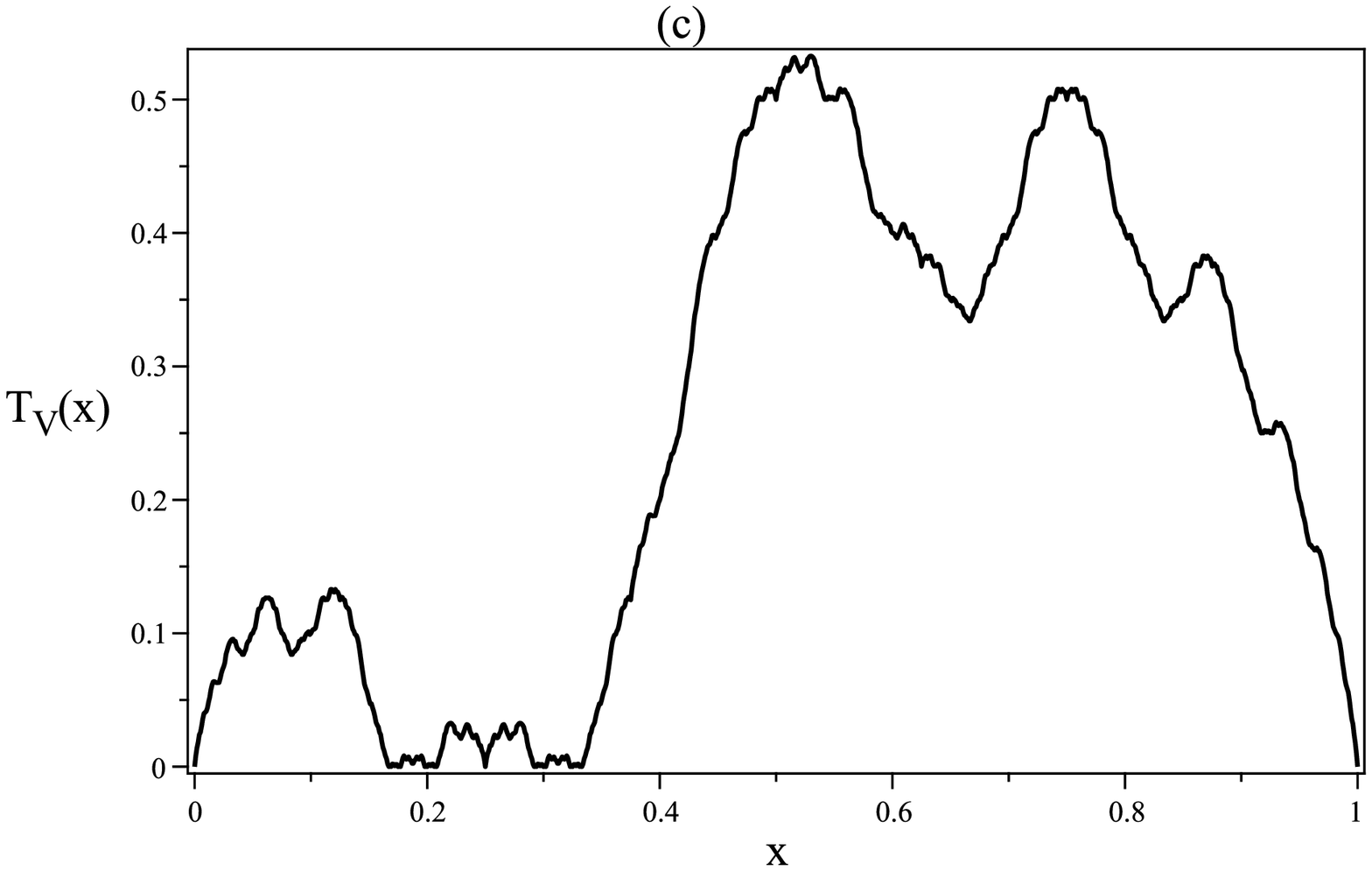} \includegraphics[width=7cm]{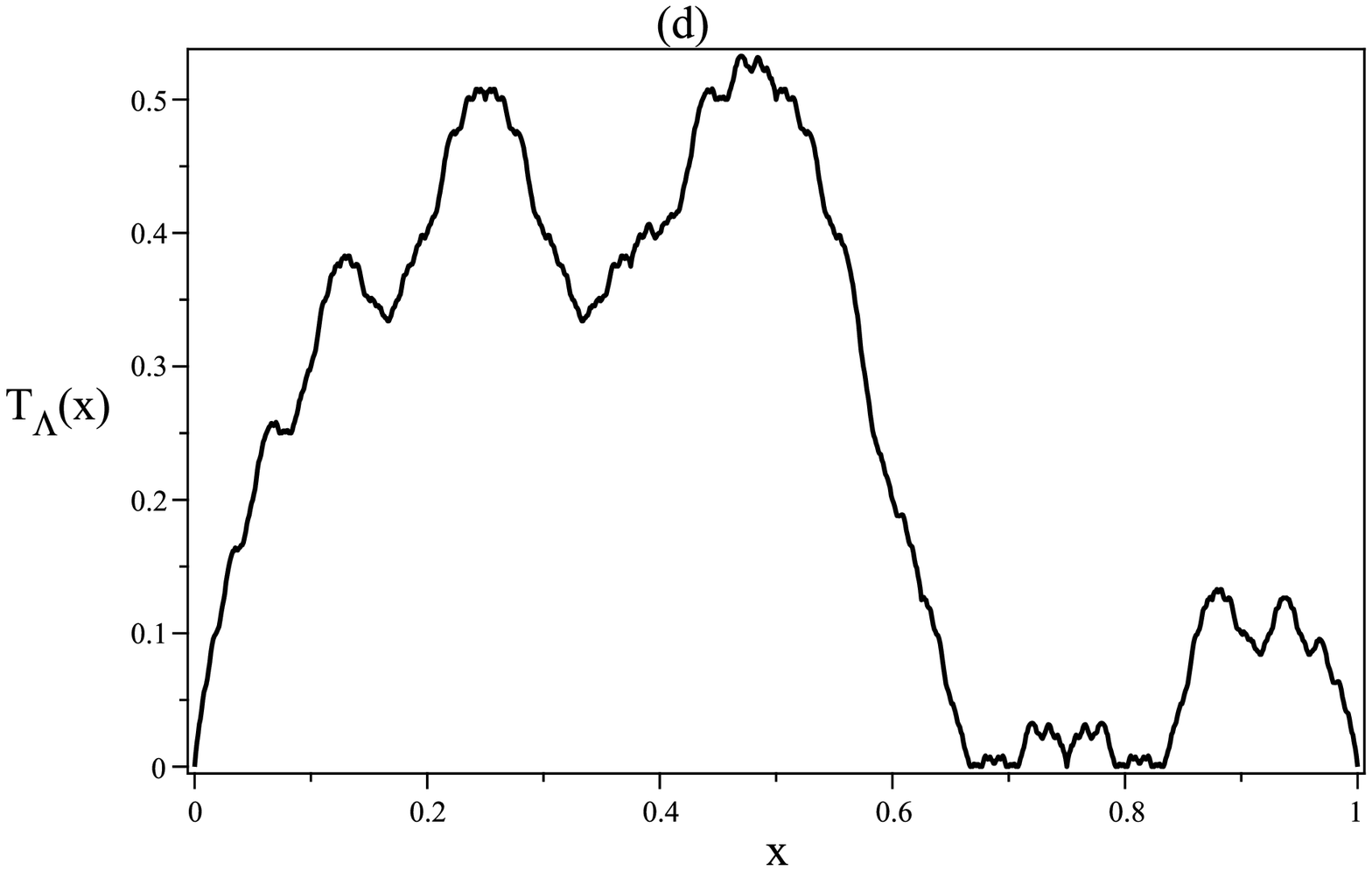}
\end{center}
\caption{\footnotesize{\emph{The Takagi functions}. In this figure, the Takagi functions are shown for the four maps at a parameter value of $h=1$. In \textbf{(a)}; the lifted Bernoulli shift, which you may recognise as the famous Takagi function. In \textbf{(b)}; the negative Bernoulli shift. In \textbf{(c)}; the lifted V map. In \textbf{(d)}; the lifted tent map. Note the self similarity and \quotemarks{fractal} structure. Note also the asymmetry in \textbf{(c)} and \textbf{(d)} compared to \textbf{(a)} and \textbf{(b)}, this is due to the asymmetry in the evolution of the p.d.f for these maps. Also portrayed here is that $T_V(x)= T_{\Lambda}(1-x)$, this result will be explained.}}
\label{fig:4_Tak}
\end{figure}


In this section we will show how to evaluate the Takagi functions and see what some of them look like in figure \ref{fig:4_Tak}. We will start with the Takagi functions for the lifted Bernoulli shift map (equation (\ref{Eq:T_hfull})). We can repeatedly apply the recursion relation in equation (\ref{Eq:T(x)_recursion1}) and obtain a sum

\begin{eqnarray}\nonumber
               T_M(x) &=& t_M(x)+ \frac{1}{2}T_M(\tilde{M}_h(x))-\frac{1}{2}T_M(h)\\ \nonumber
                      &=& t_M(x)+ \frac{1}{2}t_M(\tilde{M}_h(x)) +\frac{1}{4}T_M(\tilde{M}_h^2(x))-\frac{1}{2}T_M(h)-\frac{1}{4}T_M(h)\\
                      &=& \sum_{k=0}^{\infty}\frac{1}{2^k} t_M\left(\tilde{M}_h^k(x)\right) -T_M(h)
\label{Eq:Tak_sum_1}\:,
\end{eqnarray}

\noindent where

\begin{equation}
t_M (x):=
\left\{
\begin{array}{rl}
\left\lfloor h\right\rfloor x                                        & \ \ 0\leq x <\frac{1-\hat{h}}{2}\\
\frac{\hat{h}-1}{2} +\left\lceil h\right\rceil  x                    & \ \ \frac{1-\hat{h}}{2} \leq x < \frac{1}{2}\\
\frac{1+\hat{h}}{2} - \left\lceil h\right\rceil  x                   & \ \ \frac{1}{2} \leq x < \frac{1+\hat{h}}{2}\\
\left\lfloor h\right\rfloor -   \left\lfloor h\right\rfloor  x       & \ \ \frac{1+\hat{h}}{2}\leq x \leq 1\end{array}\right. \ .
\label{Eq:t_M(x)}
\end{equation}

\noindent We can remove the $T_M(h)$ term in equation (\ref{Eq:Tak_sum_1}) to obtain an infinite sum in terms of equation (\ref{Eq:t_M(x)})

\begin{eqnarray}\nonumber
               T_M(x)   &=& \sum_{k=0}^{\infty}\frac{1}{2^k} t_M\left(\tilde{M}_h^k(x)\right) -\frac{1}{2}\left(T_M(h)+T_M\left(h\right)\right)\\ \nonumber
                        &=& \sum_{k=0}^{\infty}\frac{1}{2^k} t_M\left(\tilde{M}_h^k(x)\right) -\frac{1}{2}T_M\left(h\right)-\frac{1}{2}\left(\sum_{k=0}^{\infty}\frac{1}{2^k} t_M\left(\tilde{M}_h^k(h)\right) -T_M(h)\right)\\
                        &=&\sum_{k=0}^{\infty}\frac{1}{2^k}\left(t_M\left(\tilde{M}_h^k(x)\right)-\frac{1}{2}t_M\left(\tilde{M}_h^k(h)\right)\right).
\label{Eq:T_M_sum}
\end{eqnarray}

The Takagi functions for the lifted negative Bernoulli shift map can be evaluated in a similar way. We use the same recursive definition to obtain

\begin{eqnarray}\nonumber
               T_W(x)&=& \sum_{k=0}^{\infty} \left(\frac{-1}{2}\right)^k t_W(x)+ \sum_{k=0}^{\infty}\left(\frac{-1}{2}\right)^kT_W(h) \\
                     &=& \sum_{k=0}^{\infty} \left(\frac{-1}{2}\right)^k t_W(x)+ \frac{1}{3}T_W(h)
\label{Eq:Tak_W_sum_1}
\end{eqnarray}

\noindent where

\begin{equation}
t_W (x):=
\left\{
\begin{array}{rl}
\left\lceil h\right\rceil  x                                                                                             & \ \ 0\leq x <\frac{\hat{h}}{2}\\
\frac{\hat{h}}{2}\left( \lceil h \rceil - \lfloor h \rfloor\right)+\left\lfloor h\right\rfloor  x                       & \ \ \frac{\hat{h}}{2} \leq x < \frac{1}{2}\\
\lfloor h\rfloor+\frac{\hat{h}}{2}\left( \lceil h \rceil - \lfloor h \rfloor\right)- \left\lfloor h\right\rfloor x      & \ \ \frac{1}{2} \leq x < 1-\frac{\hat{h}}{2}\\
\left\lceil h\right\rceil -\left\lceil h\right\rceil  x                                         & \ \ 1-\frac{\hat{h}}{2}\leq x \leq 1\end{array}\right. \ .
\label{Eq:t_W(x)}
\end{equation}

\noindent We can again remove $T_W(h)$ to obtain an infinite sum in terms of equation (\ref{Eq:t_W(x)}).

\begin{equation}
               T_W(x)   = \sum_{k=0}^{\infty}\left(\frac{-1}{2}\right)^k \left(t_W\left(\tilde{W}_h^k(x)\right) +\frac{1}{2}t_W\left(\tilde{W}_h^k(h)\right)\right).
\label{Eq:T_W_sum}
\end{equation}

Things are not so simple when we try to evaluate the Tagagi functions for the lifted V map and the lifted Tent map. The fact that we have both a positive and a negative gradient in $V_h(x)$ and $\Lambda_h(x)$ makes it harder to simplify the recursion relation into one single sum. If we take the Takagi functions for the lifted V map as an example and let

\begin{equation}
f (x):=
\left\{
\begin{array}{rl}
-1               & \ \ 0\leq x <\frac{1}{2}\\
1                & \ \ \frac{1}{2}\leq x \leq 1\end{array}\right. \ ,
\label{Eq:f(x)}
\end{equation}

\begin{equation}
g (x):=
\left\{
\begin{array}{rl}
T_V(h)           & \ \ 0\leq x <\frac{1}{2}\\
-T_V(1-h)        & \ \ \frac{1}{2}\leq x \leq 1\end{array}\right. \ ,
\label{Eq:g(x)}
\end{equation}

\noindent then the sum for the Takagi function that we obtain is

\begin{equation}
               T_V(x)= t_V(x) + \frac{1}{2}g(x)+ \sum_{k=1}^{\infty}\frac{f(\tilde{V}_h^{k-1}(x))}{2^k}\left(t_V(\tilde{V}_h^k(x))+\frac{1}{2}g(\tilde{V}_h^k(x))\right)
\label{Eq:T_V_sum}
\end{equation}

\noindent where

\begin{equation}
t_V (x):=
\left\{
\begin{array}{rl}
x                & \ \ 0\leq x <\frac{h}{2}\\
\frac{h}{2}      & \ \ \frac{h}{2} \leq x < \frac{1}{2}\\
\frac{h+1}{2} -x & \ \ \frac{1}{2} \leq x < \frac{1+h}{2}\\
0                & \ \ \frac{1+h}{2}\leq x \leq 1\end{array}\right. \ . \ \ 0 \leq h \leq 1.
\label{Eq:t_V(x)}
\end{equation}

\noindent So this difficulty is the first consequence of the asymmetry in these maps. We can however simplify the Takagi functions for the lifted V map and the lifted tent map by setting $h=1$ and hence evaluate them. See figure \ref{fig:4_Tak} for an illustration of some of these interesting functions.

\section{The structure of the diffusion coefficients}
\label{sec:struc_diff}

\begin{figure}[h]
\begin{center}
  \includegraphics[width=7cm]{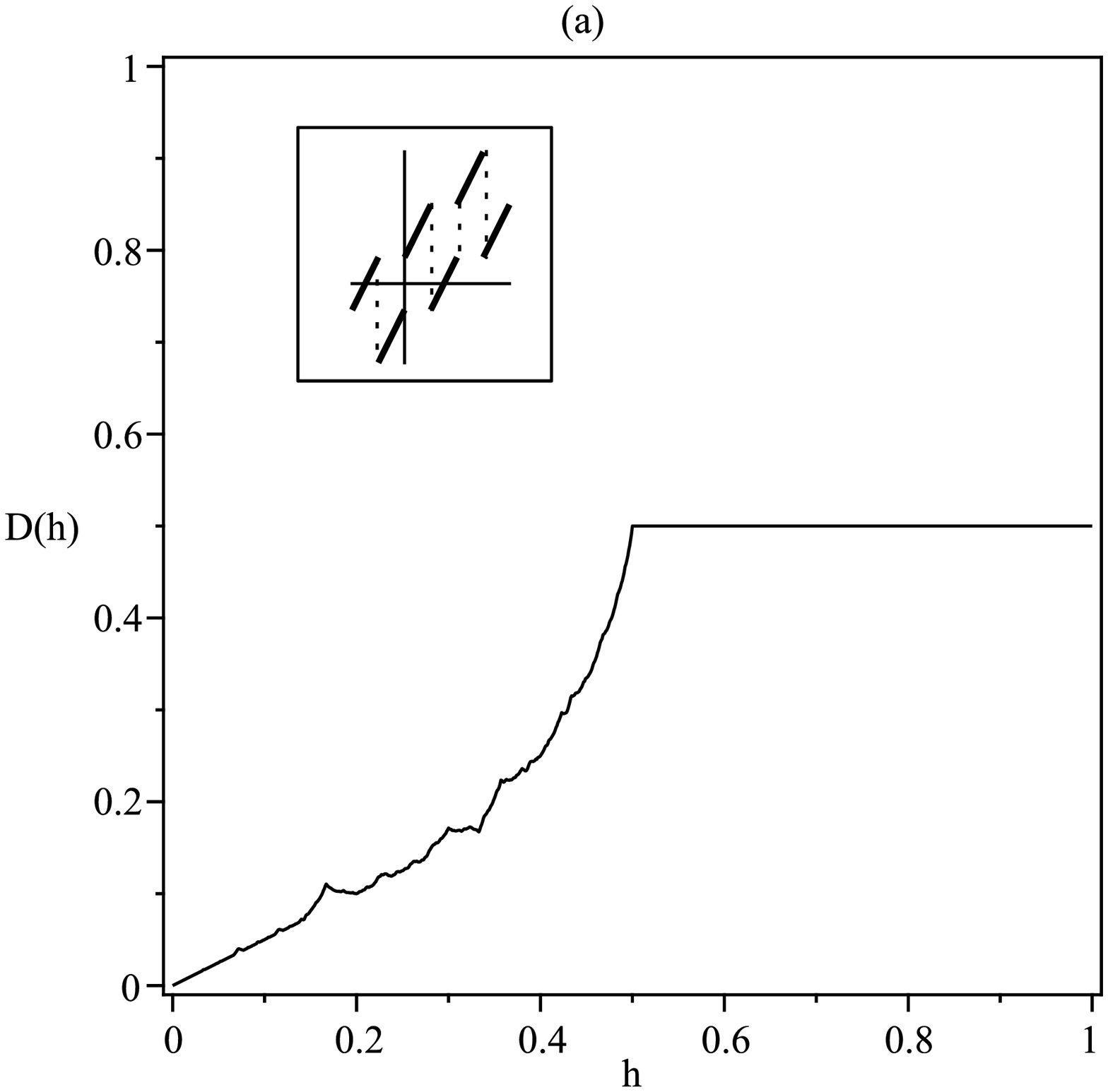} \includegraphics[width=7cm]{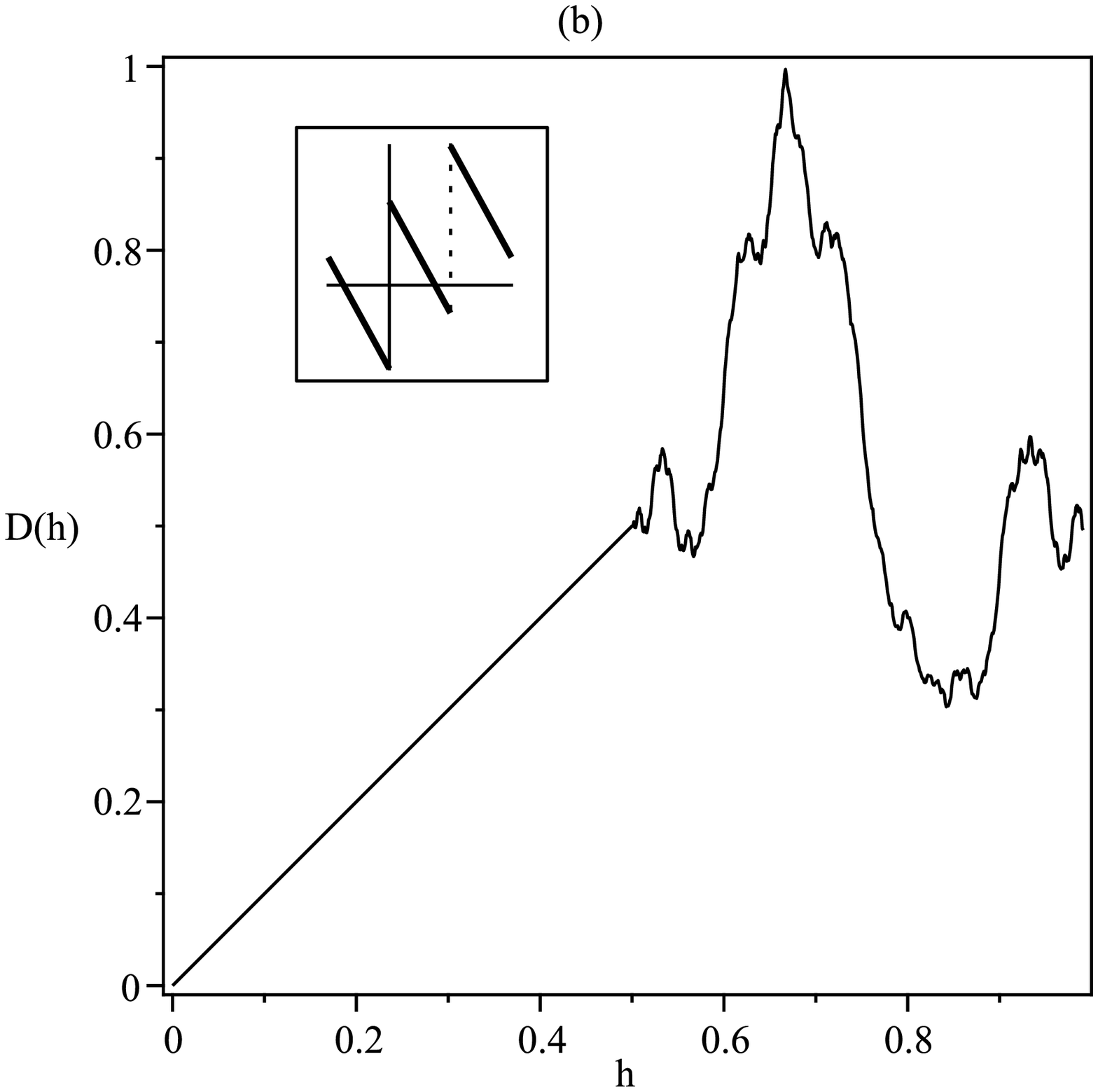}\\
  \includegraphics[width=7cm]{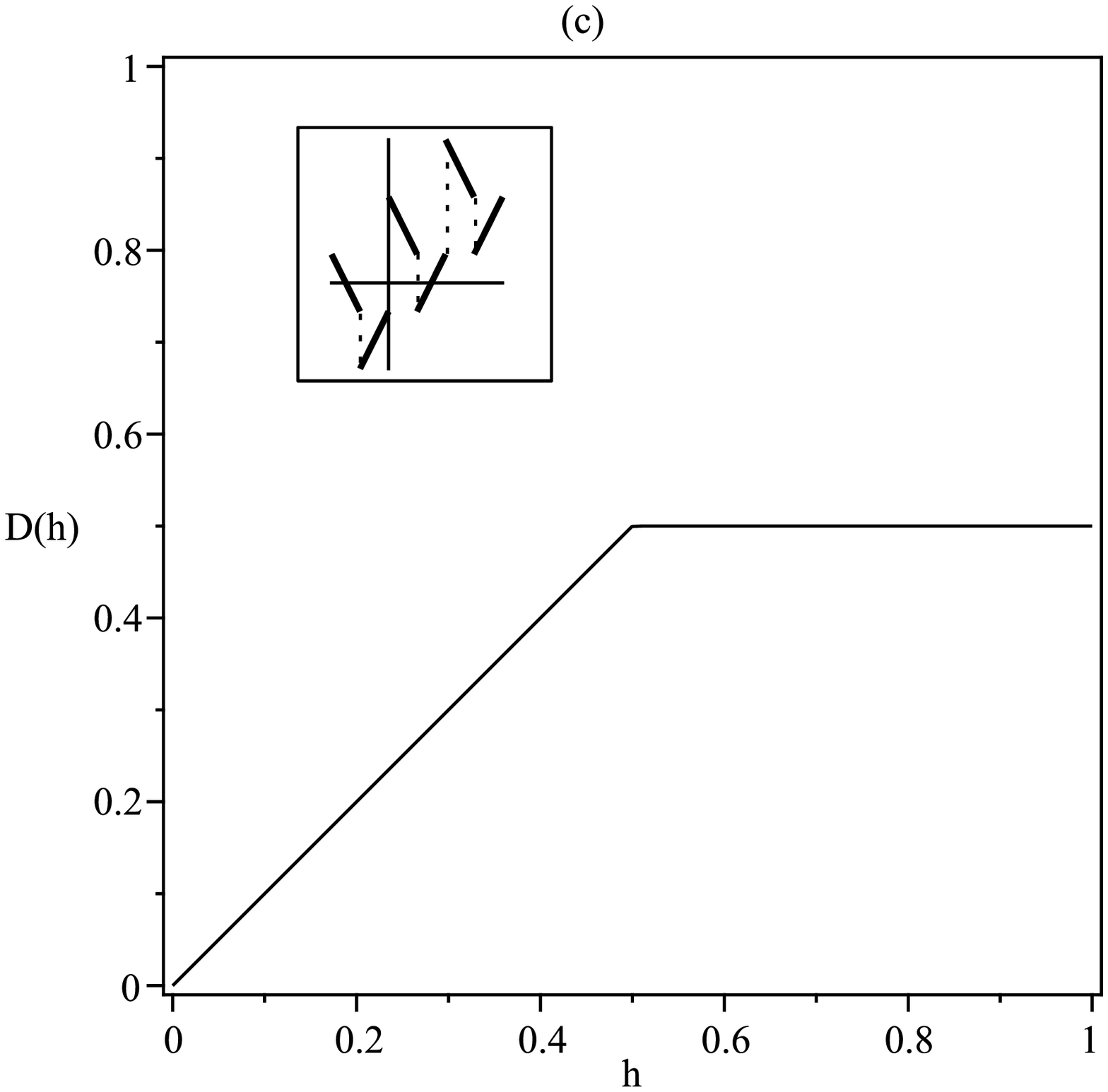} \includegraphics[width=7cm]{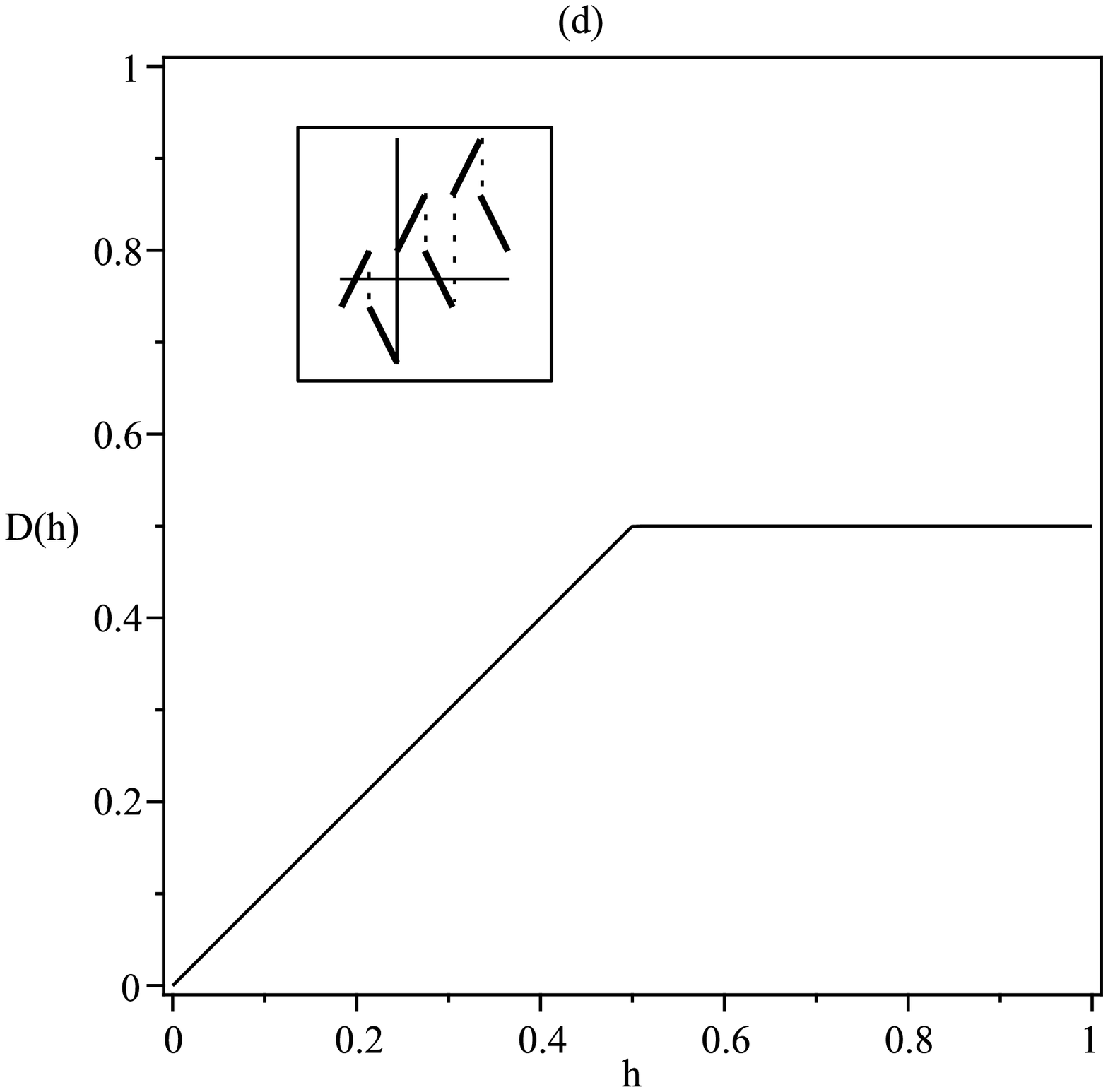}
\end{center}
\caption{\footnotesize{\emph{The diffusion coefficients}. In this figure, the parameter dependent diffusion coefficients are illustrated. In \textbf{(a)} the lifted Bernoulli shift. In \textbf{(b)} the lifted negative Bernoulli shift. In \textbf{(c)} the lifted V map. In \textbf{(d)} the lifted tent map.}}
\label{Fig:diff_coeffs}
\end{figure}

In this section the parameter dependent diffusion coefficients for the four maps will be shown, and their structure explained.

\subsection{The lifted Bernoulli shift map}
\label{subsec:liftedBern}

Figure \ref{Fig:diff_coeffs}.(a) gives the diffusion coefficient for the lifted Bernoulli shift map. The two striking features are the fractal region when $h$ is between zero and a half, and the linear plateau when $h$ is between a half and one. These regions will be explained in turn.

\subsubsection{The fractal region}
\label{sec:fractal_region}

Firstly, the term \quotemarks{fractal} has no strict mathematical
definition so we use the term loosely. In particular we use it to
refer to the fact that the diffusion coefficient exhibits non-trivial
fine scale structure, and regions of scaling and self similarity; for
a discussion of this see \cite{Klages-03}.

The topological instability of the map under parameter variation is reflected in the fractal structure of the diffusion coefficient. So in order to understand the fractality, we need to understand the topological instability. To this end, we take equation (\ref{Eq:M_h_box}) modulo $1$, and analyse the behaviour of the Markov partitions of the interval map $\tilde{M}_h(x):[0,1]\rightarrow [0,1]$

\begin{equation}
\tilde{M}_h (x)=
\left\{
\begin{array}{rl}
2x + \hat{h}     &\ \ 0 \leq x < \frac{1-\hat{h}}{2}\\
2x + \hat{h} -1  &\ \ \frac{1-\hat{h}}{2} \leq x < \frac{1}{2}\\
2x - \hat{h}     &\ \ \frac{1}{2}\leq x < \frac{1+\hat{h}}{2}\\
2x - 1 - \hat{h} &\ \ \frac{1+\hat{h}}{2} \leq x < 1 \end{array}\right. \ \ .
\label{Eq:M_h_mod1}
\end{equation}

The structure of the Markov partitions of (\ref{Eq:M_h_mod1}) varies wildly under parameter variation. The method we employ to understand the Markov partitions involves iterating the critical point $x=\frac{1}{2}$, see \cite{Klages-95,Klages-99,Klages-96}. The set of iterates of this point, along with the set of points symmetric about $x=\frac{1}{2}$, will form a set of Markov partition points for the map. Hence we call the orbit of $x=\frac{1}{2}$ a \quotemarks{generating orbit}. Furthermore, if the generating orbit is finite for a particular value of $h$, we obtain a finite Markov partition. We can then use the finite Markov partition to tell us about the diffusive properties of the map and hence the structure of the diffusion coefficient. For this purpose the following proposition is crucial.

\begin{prop}

The set of values of the parameter $h$ which give a finite Markov partition are dense in the parameter space.

\end{prop}

\begin{proof}

We show that when $h$ is rational, the generating orbit is finite. This is achieved by showing that the denominator of $h$ fixes the number of possible iterates of the generating orbit. Let $h=\frac{a}{b}$ where $a,b\in \nz$ and $a\leq b$

\begin{eqnarray}\nonumber
\tilde{M}_{\frac{a}{b}} (0.5) &=& 1 - h\\
                              &=&\frac{b-a}{b}.
\end{eqnarray}

\noindent Clearly, $ b-a \in \{0,1,2,...,b-1\}$. If $\tilde{M}_{\frac{a}{b}} (\frac{b-a}{b})$ is then evaluated, there are four possibilities due to the four branches of (\ref{Eq:M_h_mod1}). It is not hard to show that for all four possibilities

\begin{equation}
                \tilde{M}_{\frac{a}{b}}\left(\frac{b-a}{b}\right)= \frac{c}{b}, \ \ c\in \{0,1,2,...,b-1\}.
\label{Eq:M_h_induction}
\end{equation}

\noindent So for $n=1$ and $n=2$

\begin{equation}
               \tilde{M}_{\frac{a}{b}}^n \left(0.5 \right)=\frac{c}{b},
\end{equation}
\noindent with $c \in \{0,1,2,...,b-1\}.$ Now we assume that

\begin{equation}
                \tilde{M}_{\frac{a}{b}}^m \left(0.5 \right)=\frac{c}{b}
\end{equation}

\noindent for all $m\leq n$, and the result follows by induction that

\begin{equation}
                \tilde{M}_{\frac{a}{b}}^n(0.5) = \frac{c}{b}, \ \ c\in \{0,1,2,...,b-1\} \ \ \forall n\in \nz.
\label{Eq:Mar_par_res}
\end{equation}

The result in (\ref{Eq:Mar_par_res}) puts a limit on the size of the subset of values that the orbit of $x=0.5$ can hit at a given rational value of $h$. This size being equal to $|\{0,\frac{1}{b},...,\frac{b-1}{b}\}|=b$. Hence the orbit must be periodic or pre-periodic, and the Markov partition of the map must have a finite number of partition points when $h$ is rational.

\end{proof}

The second important result is that the finite Markov partitions correspond to the local minima and maxima of the diffusion coefficient \cite{Klages-95, Klages-96, Klages-99}. If we extend our view back to the full maps, we see that if the generating orbit is periodic, i.e.

\begin{equation}
                \tilde{M}_h^n(0.5)=0.5, \ \ n\in \nz,
\label{Eq:M0.5=0.5}
\end{equation}

\noindent then this corresponds to a relatively high rate of diffusion for the parameter value, which is reflected in the diffusion coefficient as a local maximum. In contrast, if the generating orbit is pre-periodic, this corresponds to a relatively low rate of diffusion for the parameter value, which is reflected in the diffusion coefficient as a local minimum. So given that we have a dense set of local maxima and minima, we observe a fractal diffusion coefficient. Furthermore, equation (\ref{Eq:M0.5=0.5}) furnishes us with the means of pinpointing the local maxima of the diffusion coefficient. In \cite{Klages-96} this technique gave an approximation of where the local extrema were, but this model allows to to find them precisely. Each $n$ gives us a set of simple linear equations for the variable $h$, the solutions of which are the local maxima in the diffusion coefficient. In addition, the smaller values of $n$ give the most striking local maxima. One can apply a similar technique to locate the local minima of the graph, see figure \ref{fig:D(h)_Ex}. However, it's not a case of finding solutions to one simple equation like (\ref{Eq:M0.5=0.5}). Rather, there are many ways to define pre-periodic orbits as opposed to defining periodic orbits. In addition, the local minima do not adhere to such a strict ordering that the local maxima do. This is due to the fact that there are two components to a pre-periodic orbit, namely the transient length and the periodic orbit length.
In summary, for $h$ between zero and one half, there exists a dense set of points which are either local maxima or local minima. Hence a fractal structure is observed.

\begin{figure}[h!]
\begin{center}
  \includegraphics[width=12cm]{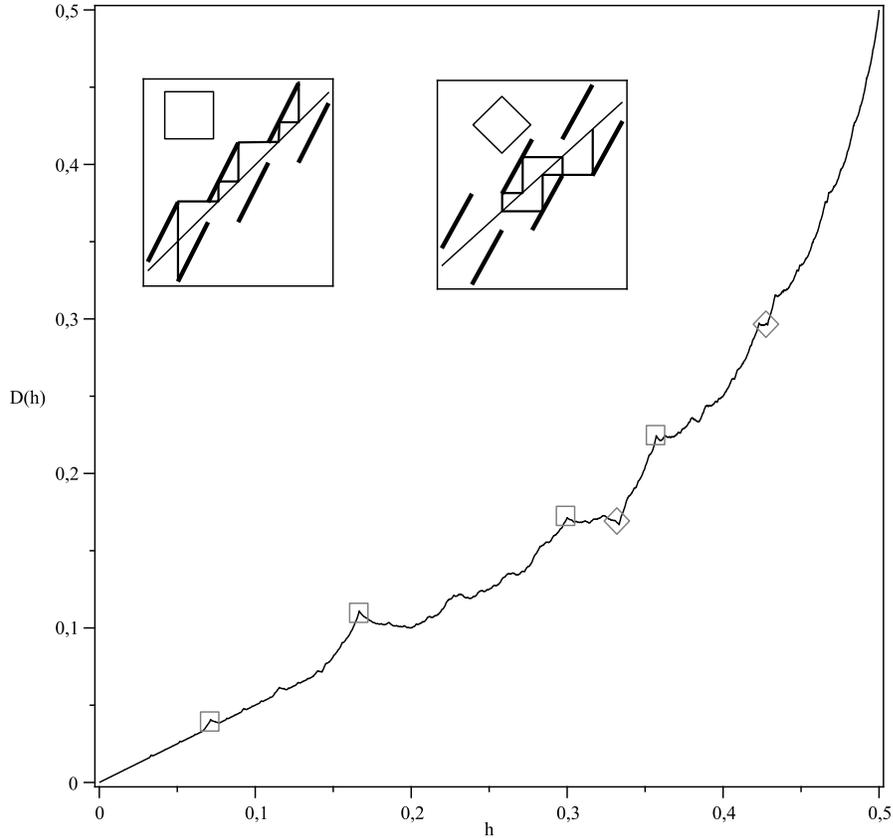}
\end{center}
\caption{\footnotesize{\emph{Pinpointing the local extrema}. In this figure, some of the local extrema have been highlighted. At the points highlighted by squares we see the orbit of $0.5$ is iterated to infinity resulting in a local maximum. At the points highlighted by diamonds we see the orbit of $0.5$ is in a closed loop resulting in a local minimum.}}
\label{fig:D(h)_Ex}
\end{figure}

\subsubsection{The linear region}
\label{sec:linear_region}

The second feature of the diffusion coefficient is the linear region where $0.5 \leq h \leq1$. Not only is it striking because it very abruptly changes from fractal to linear, it is also counter intuitive if we apply a simple random walk approximation to the map. That is, we can obtain a first order approximation of the diffusion coefficient by looking at the measure of the escape region of the map (the area where a point can move from one unit interval to the next) \cite{Klages-96,Klages-97}. This region clearly increases linearly with the parameter, so based on a first order approximation one would expect to see a general increase in the diffusion coefficient as the parameter increases. However, what is observed defies this. We explain this feature by noting the non-ergodicity of the map.

\begin{figure}[h]
\includegraphics[width=7cm]{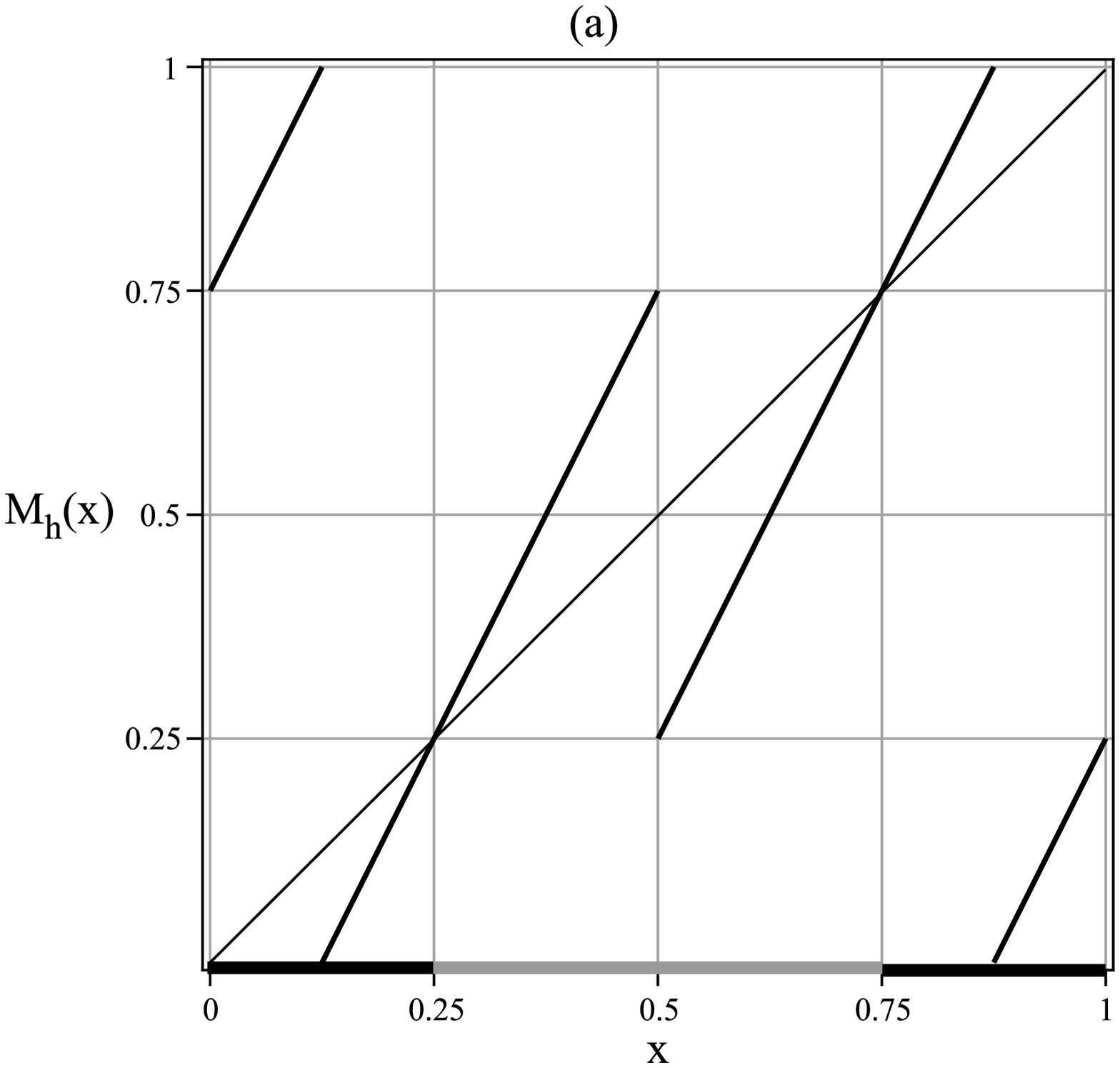}  \includegraphics[width=7cm]{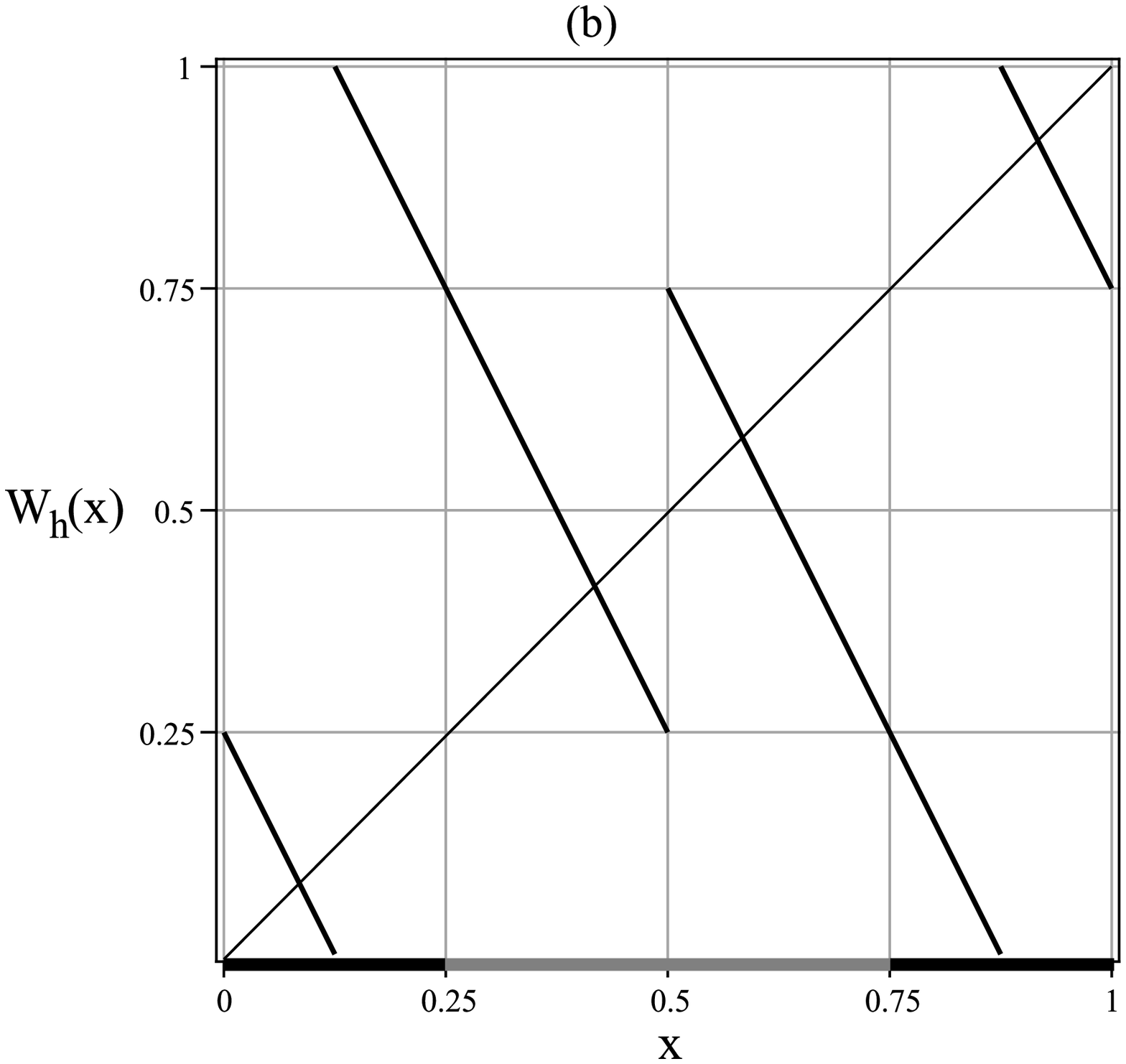}
\caption{\footnotesize{\emph{Nonergodicity}. In this figure, the nonergodicity of the the lifted Bernoulli shift map at $h=0.75$ is shown in (a), and in the lifted negative Bernoulli shift at $h=0.25$ in (b). For simplicity the dynamics have been reduced to the maps modulo $1$. One can see that the black areas get mapped into themselves as do the grey areas. This splits up the phase space breaking ergodicity.}}
\label{Fig:non_erg}
\end{figure}

When the parameter $h$ reaches one half, a fixed point is born in the modulo $1$ map. As the parameter increases further, the fixed point bifurcates and the two resulting fixed points split the phase space up into two invariant sets, breaking the ergodicity of the map, see figure \ref{Fig:non_erg}. Consequently, the invariant density $\rho^*(x)$, can be interpreted as the sum of two invariant densities $\rho^*_1(x)$ and $\rho^*_2(x)$ and the diffusion coefficient can be evaluated as

\begin{eqnarray}\nonumber
D_M(h)                              &=&       \lim_{n\to\infty}\frac{1}{2n}\int_0^1 \rho^*(x)(x_n-x_0)^{2}dx\\
                                  &=&       \lim_{n\to\infty}\frac{1}{2n}\left(\int_0^1 \rho_1^*(x)(x_n-x_0)^{2}dx + \int_0^1 \rho_2^*(x)(x_n-x_0)^{2} dx\right).
\label{Eq:D(h)_non_erg}
\end{eqnarray}

\noindent We can then use the Taylor-Green-Kubo formula and derive two separate generalised Takagi functions, in order to evaluate the diffusion coefficient as

\begin{eqnarray}\nonumber
                         D_M(h)      &=&     (1-h) + (h-\frac{1}{2})\\
                                     &=&      \frac{1}{2}.
\label{Eq:D(h)_non_erg_2}
\end{eqnarray}


\subsection{The negative Bernoulli shift}
\label{sec:neg_bern_shift}

Figure \ref{Fig:diff_coeffs}(b) shows the diffusion coefficient for the lifted negative Bernoulli shift map $W_h(x)$. Firstly we note how radically different the structure of the diffusion coefficient is from the lifted Bernoulli shift map. In addition, as $h \rightarrow 0$ the diffusion coefficient does not go to $\frac{h}{2}$ by reproducing the random walk solution as in the lifted Bernoulli shift map. However, as $h \rightarrow \infty$  the diffusion coefficient goes to $\frac{h^2}{2}$ as in the lifted Bernoulli shift. We also observe both a linear and a fractal region. For $0\leq h \leq\frac{1}{2}$ the diffusion coefficient is a simple linear function equal to $h$. Again, the explanation for this is that the map is non-ergodic in this parameter range, see figure \ref{Fig:non_erg}. The phase space is split up into two invariant regions, one of which does not contribute to diffusion as it is a trapping region, the other of which grows linearly with $h$ so we observe a linear, increasing diffusion coefficient. When $\frac{1}{2}\leq h \leq 1$, the map becomes topologically unstable under parameter variation. Similar to the lifted Bernoulli shift map, this instability is reflected in the behaviour of the Markov partitions of (\ref{Eq:V_h_box}) taken modulo $1$. We have that the finite Markov partitions are dense in the parameter space and that these finite Markov partitions correspond to the local maxima and minima of the diffusion coefficient. Hence we observe a fractal diffusion coefficient.


\subsection{The lifted V map}
\label{subsec:vmap}

Mercifully, we need not turn to infinite sums like equation (\ref{Eq:T_V_sum}) when we evaluate the diffusion coefficient for the lifted V map. The diffusion coefficient for the V-map is given by

\begin{equation}
                D_V(h)=\frac{h}{2}+\frac{1}{2}\left(T_V(h)+T_V(1-h)\right),\\ \\ 0\leq h \leq 1.
\label{Eq:D(h)_tent_v}
\end{equation}

\noindent Instead of evaluating the relevant Takagi function, equation (\ref{Eq:T(h)_V}), numerically we can use the helpful property that for $h$ less than one half

\begin{eqnarray}\nonumber
                T_V(h)&=& -\frac{1}{2}T_V(1-h)+\frac{h}{2}+\frac{1}{2}T_V(h)\\
                &=&-T_V(1-h)+h.
\label{Eq:eval_d_tent}
\end{eqnarray}

\noindent Using equation (\ref{Eq:eval_d_tent}) in equation (\ref{Eq:D(h)_tent_v}), the diffusion coefficient can be evaluated to

\begin{eqnarray}\nonumber
                         D_V(h)&=&\frac{h}{2}+\frac{1}{2}\left(-T_V(1-h)+h+T_V(1-h)\right)\\
                             &=& h.
\label{Eq:eval_d_tent_2}
\end{eqnarray}

\noindent Furthermore, for $h$ greater than one half

\begin{eqnarray}\nonumber
                T_V(h)&=& -h+\frac{1}{2}T_V(h)+\frac{h}{2}+\frac{1}{2}-\frac{1}{2}T_V(1-h)\\
                &=&-T_V(1-h)-h+1.
\label{Eq:eval_d_tent_3}
\end{eqnarray}

\noindent Using equation (\ref{Eq:eval_d_tent_3}) in equation (\ref{Eq:D(h)_tent_v}), the diffusion coefficient can again be evaluated

\begin{eqnarray}\nonumber
                D_V(h)&=&\frac{h}{2}+\frac{1}{2}\left(-T_V(1-h)-h+1+T_V(1-h)\right)\\
                    &=& \frac{1}{2}.
\label{119}
\end{eqnarray}

\noindent See figure \ref{Fig:diff_coeffs}.(c) for an illustration.

\subsection{The lifted tent map}
\label{subsec:tent}

For the lifted tent map $\Lambda_h(x)$, we can not perform the same trick with the Takagi functions, equation (\ref{Eq:T(h)_tent}), that we did with the lifted V map in subsection \ref{subsec:vmap}. However we still do not need to resort to numerical computations to see what the diffusion coefficient looks like as we can note that

\begin{equation}
                \Lambda_h(x)=-V_h(-x).
\label{Eq:conjugacy}
\end{equation}

\noindent Equation (\ref{Eq:conjugacy}) is important because it serves as a topological conjugacy of the form $f(x)=-x$. By using the Taylor-Green-Kubo formula equation (\ref{Eq:TGK2}) it was shown in \cite{Korabel-2004} that the diffusion coefficient is preserved under topological conjugacy, hence we observe an identical diffusion coefficient in the two maps which can be seen in figure \ref{Fig:diff_coeffs}. An alternative way is to start from Einstein's formula equation (\ref{Eq:Einstein_diff}), as is demonstrated by proving the following proposition:

\begin{prop}

The diffusion coefficients for the lifted tent map is identical to that of the lifted V map.

\end{prop}

\begin{proof}

We note that the diffusion coefficient for the lifted V map is given by

\begin{equation}
               D_V(h)=\lim_{n\rightarrow\infty}\frac{\left\langle \left(    V_h^n(x_0)-x_0     \right)^2   \right\rangle}{2n},
\label{Eq:D_v}
\end{equation}

\noindent and that the diffusion coefficient for the tent map ($D_{\Lambda}(h)$) can be given by

\begin{equation}
               D_\Lambda(h)=\lim_{n\rightarrow\infty}\frac{\left\langle \left(\Lambda^n_h(-x_0)-(-x_0)     \right)^2   \right\rangle}{2n},
\label{Eq:D_Lambda}
\end{equation}

\noindent where the p.d.f is in the interval, $[-1,1]$ to make it symmetric about $x=0$. By substituting (\ref{Eq:conjugacy}) into (\ref{Eq:D_Lambda})

\begin{eqnarray}\nonumber
               D_\Lambda(h) &=& \lim_{n\rightarrow\infty}\frac{\left\langle \left(\Lambda^n_h(-x_0)-(-x_0)     \right)^2   \right\rangle}{2n}\\ \nonumber
                            &=& \lim_{n\rightarrow\infty}\frac{\left\langle \left(-V^n_h(x_0)-(-x_0)     \right)^2   \right\rangle}{2n}\\ \nonumber
                            &=& \lim_{n\rightarrow\infty}\frac{\left\langle \left(V^n_h(x_0)-x_0     \right)^2   \right\rangle}{2n}\\
                            &=& D_V(h).
\label{Eq:D_l=D_v}
\end{eqnarray}

\noindent we arrive at the desired relationship that $D_\Lambda(h)=D_V(h)$.

\end{proof}

From what we have seen in the lifted Bernoulli shift map and the lifted negative Bernoulli shift map, we would expect to find non-ergodicity in the lifted tent and lifted V maps. Furthermore we would expect to find it across the entire parameter range given the linear diffusion coefficients. However, there is no obvious non-ergodicity, and although a proof that the maps are ergodic across the entire parameter range remains elusive, we can check for ergodicity at individual values of the parameter by checking the reducibility of the transition matrices \cite{Petersen-1983}. So we can confirm ergodicity for some values of $h$. Furthermore, the behaviour under parameter variation of the Markov partitions of the maps (\ref{Eq:V_h_box}) and (\ref{Eq:Lambda_h_box}) taken modulo $1$ indicate the presence of topological instability in the parameter space. So, although all the ingredients for a fractal diffusion coefficient are present, we observe a linear one.

In order to understand the linearity of the diffusion coefficient of the lifted V map (and hence the lifted tent map also), we first note the similarity of the linear regions in the diffusion coefficients of the two lifted Bernoulli shift maps, see figure \ref{Fig:diff_coeffs}. This presents the question, why do these maps have the same diffusion coefficients when they have such different microscopic dynamics?

We will explain the linearity of the diffusion coefficients of the lifted tent and lifted V map by showing why they have the same diffusion coefficients as the two Bernoulli shift maps in the relevant parameter ranges. We will take these ranges in turn starting with $0.5 \leq h \leq 1$. The diffusion coefficient for the lifted Bernoulli shift map is given by equation (\ref{Eq:D(h)_der_3}) which for $h \in [0.5,1]$ simplifies to

\begin{equation}
                 D_M(h)= \frac{h}{2}+T_M(h)
\label{Eq:D(h)_0.5_1}
\end{equation}

\noindent which can be rewritten as

\begin{equation}
                D_M(h)= \frac{h}{2} + \lim_{n\rightarrow \infty}\left( \left(\sum_{k=0}^{n-1} \frac{t_M\left(h \right)}{2^{k+1}}\right) + \frac{t_M(h)}{2^{n}} \right)
\label{Eq:D(h)_0.5_1_final}
\end{equation}

\noindent See Appendix (\ref{App:simplify}) for a full derivation of equation (\ref{Eq:D(h)_0.5_1_final}) which tells us how the diffusion coefficient converges as $n \rightarrow \infty$. Now, keeping equation (\ref{Eq:D(h)_0.5_1_final}) in mind, we turn our attention to the lifted V map. The diffusion coefficient for the lifted V map with $h \in [0,1]$ is given by

\begin{equation}
                D_V(h)= \frac{h}{2} + \lim_{n\rightarrow\infty} \left( \frac{1}{2}T^{n-1}_V(h) +\frac{1}{2}T^{n-1}_V(1-h)\right).
\label{Eq:D(h)_V_0.5_1}
\end{equation}

\noindent From equation (\ref{Eq:T(h)_V}) we can derive the useful recursion relation

\begin{equation}
               T_V^n(h)=t_V(h)+\frac{1}{2}T_V^{n-1}(h)-\frac{1}{2}T_V^{n-1}(1-h)
\label{Eq:T(h)_recurs_V}
\end{equation}

\noindent which if we continue to apply leads to

\begin{equation}
                T_V^n(h) = \sum_{k=0}^n \frac{1}{2^k} t_V(h) - \sum_{k=1}^n \frac{1}{2^k}T_V^{n-k}(1-h).
\label{Eq:TV_rec_app}
\end{equation}

\noindent Using equation (\ref{Eq:TV_rec_app}) in equation (\ref{Eq:D(h)_V_0.5_1})

\begin{equation}
               D_V(h)=\frac{h}{2} + \lim_{n\rightarrow \infty} \left(\sum_{k=0}^{n-1} \frac{t_V(h)}{2^{k+1}}+\frac{1}{2}\left( T_V^{n-1}(1-h)- \sum_{k=1}^{n-1}\frac{1}{2^{k}}T_V^{n-1-k}(1-h)\right)\right).
\label{Eq:D(h)_0.5_1_V_final}
\end{equation}

\noindent Now, we note that

\begin{equation}
                 T_V^{n-1}(1-h) = \sum_{k=1}^{n-1}\frac{1}{2^k}T_V^{n-1-k}(1-h), \ \ n \rightarrow \infty.
\label{Eq:T_cancel}
\end{equation}

\noindent so these terms cancel each other out in the limit and we are left with

\begin{equation}
               D_V(h)=\frac{h}{2}+t_V(h).
\label{Eq:DV(h)_limit}
\end{equation}

\noindent If we let $n \rightarrow \infty $ in equation \ref{Eq:D(h)_0.5_1_final} then

\begin{equation}
               \frac{t_M(h)}{2^{n}} \rightarrow 0 \ \ n \rightarrow \infty
\label{Eq:t_M(h)_0}
\end{equation}

\noindent and hence we are left with the expression

\begin{equation}
               D_M(h)=\frac{h}{2}+t_M(h).
\label{Eq:DM(h)_limit}
\end{equation}

\noindent We can see from equation (\ref{Eq:t_M(x)}) and equation (\ref{Eq:t_V(x)}) that $t_M(h)=t_V(h)$ when $h \in [0.5,1]$ and hence the two diffusion coefficients are equal despite $M_h(x)$ being non-ergodic and $V_h(x)$ ergodic in this parameter range. However the two equations (\ref{Eq:D(h)_0.5_1_V_final}) and (\ref{Eq:D(h)_0.5_1_final}) tell us that $D_M(h)$ and $D_V(h)$ converge at different rates, see figure (\ref{Fig:time_dep_diff}) for an illustration of this phenomenon.We also note that the diffusion coefficient is only dependent on $t_M(h)$ or $t_V(h)$ and that these functions are only dependent on the branch of the map in $[0.5,1]$. We interpret this phenomenon physically as the diffusion undergoing a dominating branch process, i.e. the diffusion coefficient is only dependent on the contribution of one branch of the map in the limit as $n \rightarrow \infty$. Hence we see identical diffusion coefficients despite the different microscopic dynamics.

We find a similar situation for $h \in [0,0.5]$, where the lifted negative Bernoulli shift map and the lifted V map have the same diffusion coefficient. They have the branch in $[0,0.5]$ in common in this parameter range and it is this which creates the dominating branch process. Hence we also observe identical diffusion coefficients between these two maps in this parameter range.

\begin{figure}[h!]
\hspace{-0.3cm} \includegraphics[width=7.5cm]{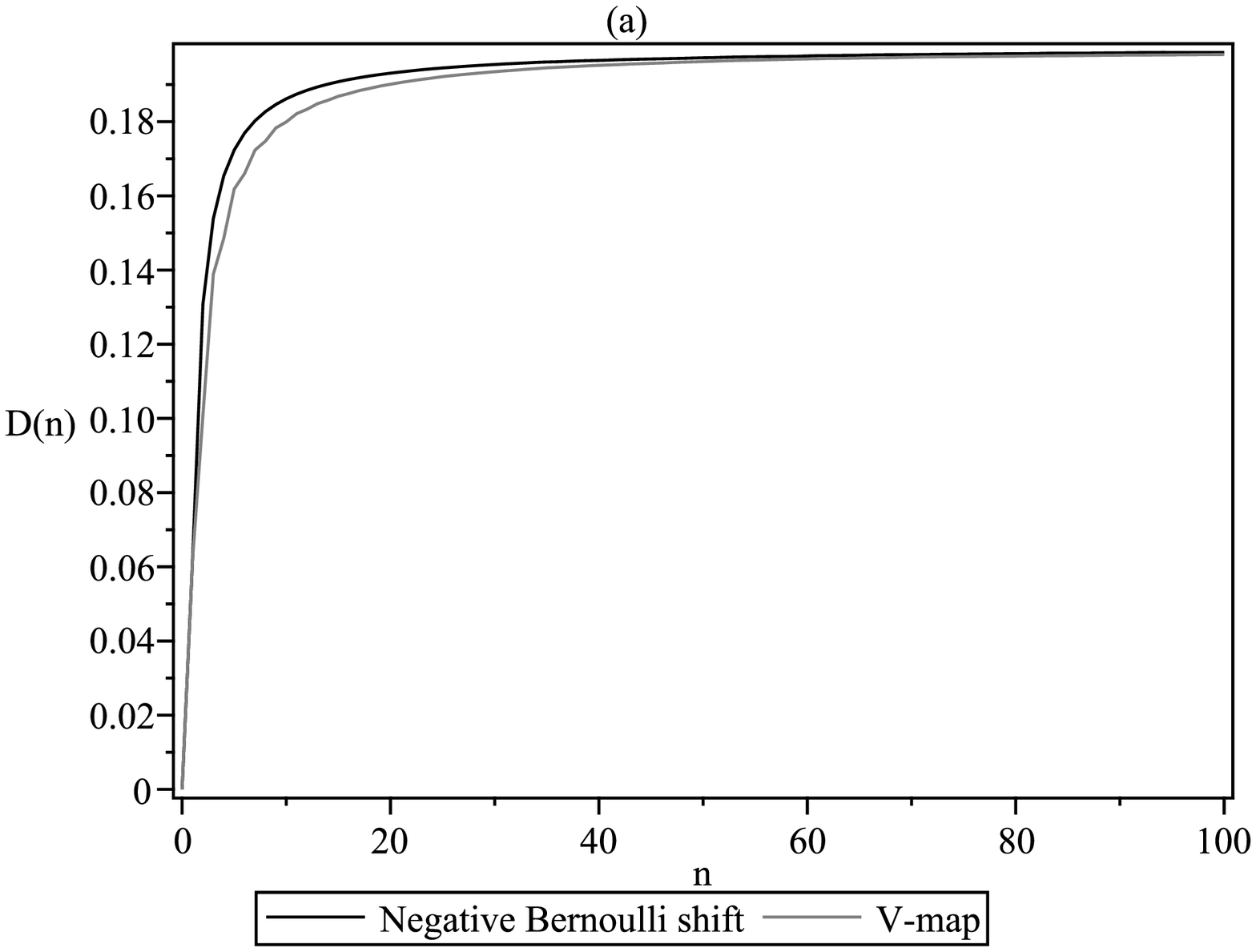}\nobreak \includegraphics[width=7.5cm]{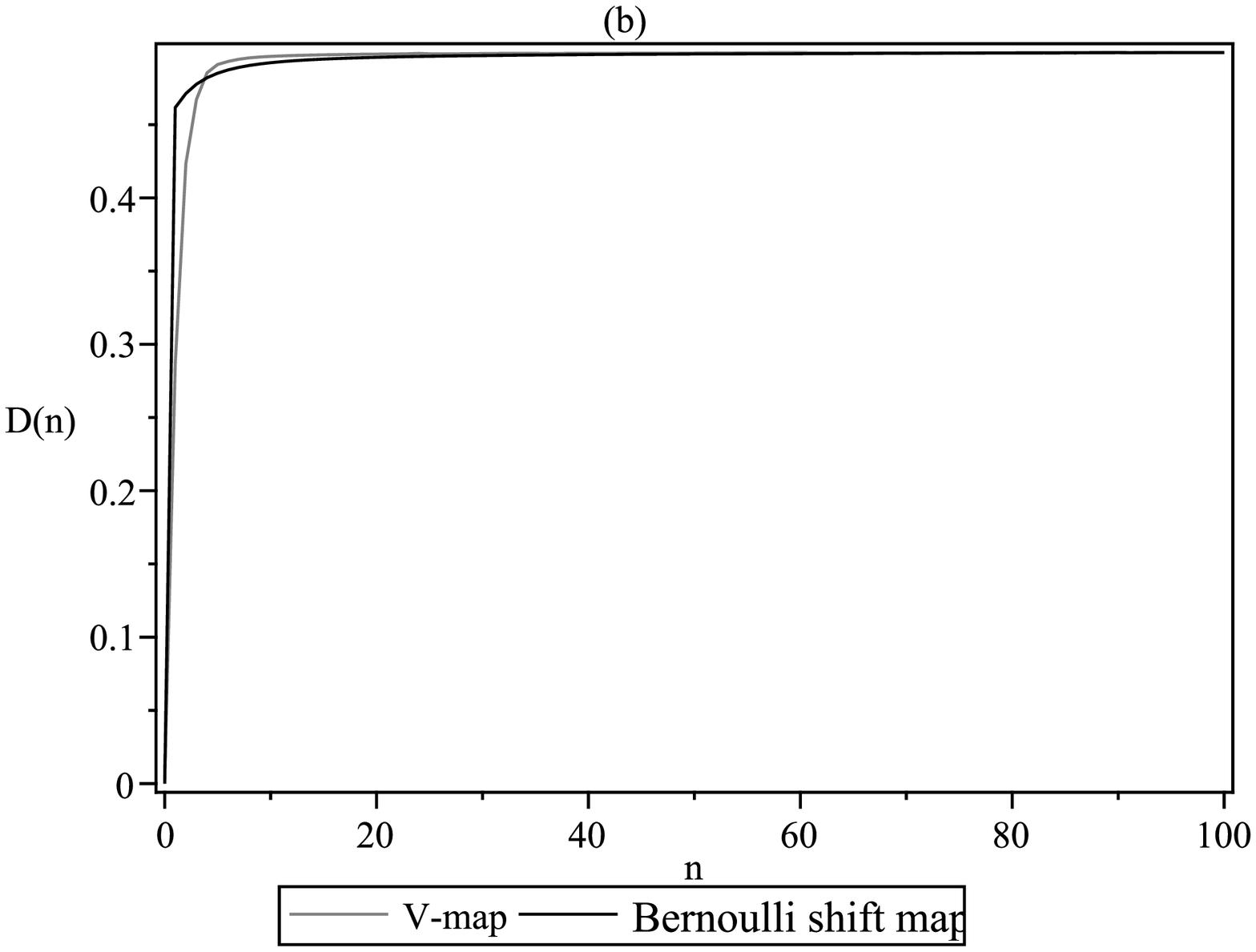}
\caption{\footnotesize{\emph{The time dependent diffusion coefficients}. In this figure we see how the diffusion coefficient converges for certain parameter values. In (a) $h=0.2$ for the lifted negative Bernoulli shift map and the lifted V map,  in (b) $h=0.7$ for the lifted Bernoulli shift map and the V map. We see that the diffusion coefficients tend to the same value but at different rates, indicating that the difference in the microscopic dynamics does not play a role in the limit. Rather we observe a dominating branch process where the common branch of the map determines the diffusion coefficient.}}
\label{Fig:time_dep_diff}
\end{figure}

\subsection{Stability of D(h) in the non-ergodic regions}
\label{subsec:stability}

In this subsection we look at the two lifted Bernoulli shift maps in the non-ergodic regions. We have already seen that changing the gradient of one branch of the map (resulting in the lifted V map or lifted tent map) has no effect on the diffusion coefficient in these regions, even though the microscopic dynamics are affected greatly. We will explore this phenomenon further.

For $h \in [0,0.5]$ the lifted negative Bernoulli shift map has a linear diffusion coefficient. By changing the gradient and chopping up the second branch of the map we obtain a map $\hat{W}_h(x):[0,1]\rightarrow \rz$

\begin{equation}
\hat{W}_h (x)=
\left\{
\begin{array}{rl}
-2x +1 + h & 0\leq x <\frac{1}{2}\\
2x-1       & \frac{1}{2}  \leq x \leq 1 - \frac{h}{2}\\
2x-2   & 1-\frac{h}{2}\leq x < 1\end{array}\right. .
\label{Eq:W_hat_box}
\end{equation}

\noindent The Takagi function for this map with $h \in [0,1]$ is

\begin{equation}
T_{\hat{W}} (x)=
\left\{
\begin{array}{ll}
-\frac{1}{2}T_{\hat{W}}(-2x+h)+x+\frac{1}{2}T_{\hat{W}}(h)                                     &\ \ 0 \leq x < \frac{h}{2}\\
-\frac{1}{2}T_{\hat{W}}(-2x+1+h)+ \frac{h}{2}+\frac{1}{2}T_{\hat{W}}(h)                        &\ \ \frac{h}{2} \leq x < \frac{1}{2}\\
\frac{1}{2}T_{\hat{W}}(2x-1)+\frac{h}{2}                                                       &\ \ \frac{1}{2}\leq x < 1-\frac{h}{2}\\
\frac{1}{2}T_{\hat{W}}(2x-1)+1-x                                                              &\ \ 1-\frac{h}{2} \leq x < 1 \end{array}\right.
\label{Eq:T(W)_Hat}
\end{equation}

\noindent and the diffusion coefficient can be evaluated as

\begin{equation}
                D_{\hat{W}}(h)=\frac{h}{2}+\frac{1}{2}\left(T_{\hat{W}}(h)+T_{\hat{W}}(1-h)\right),\\ \\ 0\leq h \leq 1.
\label{Eq:D(h)_W_hat}
\end{equation}

\noindent Again for $h \in [0,0.5]$ we have that

\begin{equation}
                T_{\hat{W}}(h) = -T_{\hat{W}}(1-h)+h.
\label{Eq:T_relation}
\end{equation}

\noindent which implies that for $h \in [0,0.5]$ the diffusion coefficient is equal to $h$.

We can play a similar game with the lifted Bernoulli shift map which has a linear diffusion coefficient for $h \in [0.5,1]$. We can change the gradient and chop up the first branch of the map to create $\hat{M}_h(x):[0,1]\rightarrow \rz$

\begin{equation}
\hat{M}_h (x)=
\left\{
\begin{array}{rl}
-2x +1   & 0\leq x <\frac{1-h}{2}\\
-2x +2   & \frac{1-h}{2} \leq x \leq \frac{1}{2}\\
2x-1-h   & \frac{1}{2}\leq x < 1\end{array}\right. .
\label{Eq:M_hat_box}
\end{equation}

\noindent The Takagi function for this map with $h \in [0,1]$ is

\begin{equation}
T_{\hat{M}} (x)=
\left\{
\begin{array}{ll}
-\frac{1}{2}T_{\hat{M}}(-2x+1)                                                 &\ \ 0 \leq x < \frac{1-h}{2}\\
-\frac{1}{2}T_{\hat{M}}(-2x+1)+x- \frac{1-h}{2}                                &\ \ \frac{1-h}{2} \leq x < \frac{1}{2}\\
\frac{1}{2}T_{\hat{M}}(2x-h)-x+\frac{1+h}{2} -\frac{1}{2}T_{\hat{M}}(1-h)      &\ \ \frac{1}{2}\leq x < \frac{1+h}{2}\\
\frac{1}{2}T_{\hat{M}}(2x-1-h)-\frac{1}{2}T_{\hat{M}}(1-h)                     &\ \ \frac{1+h}{2} \leq x < 1 \end{array}\right.
\label{Eq:T(h)_Hat}
\end{equation}

\noindent and the diffusion coefficient for this map is

\begin{equation}
                D_{\hat{M}}(h)=\frac{h}{2}+\frac{1}{2}\left(T_{\hat{M}}(h)+T_{\hat{M}}(1-h)\right),\\ \\ 0\leq h \leq 1.
\label{Eq:D(h)_M_hat}
\end{equation}

\noindent We see that for $h\in [0.5,1]$ we have

\begin{equation}
                T_{\hat{M}}(h)= -T_{\hat{M}}(1-h)-h+1
\label{Eq:T_relation_2}
\end{equation}

\noindent which implies that for $h \in [0.5,1]$ the diffusion coefficient is equal to $0.5$.

So we have again seen that the diffusion coefficients for these maps are very stable in the relevant parameter ranges, i.e. the ranges where the maps are non-ergodic. As long as the diffusion coefficient is given by

\begin{equation}
               D(h)= \frac{h}{2}+ \frac{1}{2}T(h)+\frac{1}{2}T(1-h),
\label{Eq:D(h)_gen}
\end{equation}

\noindent we can manipulate one branch of the Bernoulli-shift or the negative-Bernoulli-shift maps and the diffusion coefficient will remain unaffected in the non-ergodic parameter ranges, despite the fact that the non-ergodicity may be broken.

\section{Conclusion}
\label{sec:conclusion}

We have derived exact analytical expressions for the parameter dependent diffusion coefficients of four one dimensional maps. This was achieved by using Taylor-Green-Kubo formulae and generalised Takagi functions. Under parameter variation we have observed both fractal and linear behaviour in the diffusion coefficients. The fractality was explained in terms of the topological instability of the maps under parameter variation and this was understood by analysing the Markov partitions of the map. The linearity was explained in terms of the non-ergodicity of the maps in certain parameter ranges, this non-ergodicity splits the phase space up into two ergodic components each with their own diffusion coefficient. These individual diffusion coefficients compliment each other to create a linear diffusion coefficient. We also observed linear diffusion coefficients despite all the hallmarks of fractality being present like topological instability under parameter variation \cite{Klages-07,Klages-95,Klages-99,Klages-96} and ergodicity. In this case we found that in the relevant parameter range, the ergodic maps have a set of branches in common with the non-ergodic maps, these common branches dominate the diffusion process in the long time limit and hence we observe identical diffusion coefficients. In addition, when the parameter causes these maps to be non-ergodic, we found that the diffusion coefficients of these maps are so stable that we can drastically alter the microscopic dynamics without affecting the diffusion coefficient. This finding serves as a counter example to the previously held belief that if your system was topologically unstable and ergodic, you would observe a fractal diffusion coefficient. It is also a counter example to the belief that if you have a linear diffusion coefficient then you have a topologically stable system.

Future work will involve finding out under exactly what conditions we can manipulate the microscopic dynamics of these maps and still observe the same diffusion coefficient. We would also like to learn whether there exist any other systems which display this dominating branch phenomenon. We could potentially apply the techniques used here to more realistic, higher dimensional systems and see if we still obtain similar results. Also of interest is the consequences of introducing a bias into the system generating a current. It would be worthwile to study whether analogous phenomena exist for this other transport property and whether they can be revealed by similar techniques.

\section*{Acknowledgements}
\label{sec:Acknowlegdgements}

The authors would like to thank Gerhard Keller for helpful discussions about this work. They would like to dedicate this article to the memory of Professor Shuichi Tasaki, a pioneer in the field of dynamical systems theory applied to nonequilibrium statistical mechanics.

\appendix

\section{Takgagi function recursion relation for $M_h(x)$}
\label{App:simplify}

We can modify the Takagi function recursion relation for the lifted Bernoulli shift map in order to understand how the diffusion coefficient converges. We restrict the parameter $h$ to $[0,1]$ and see from equation (\ref{Eq:T_hfull}) that

\begin{eqnarray}\nonumber
                T^n_M(h) &=& t_M(h) + \frac{1}{2}T_M^{n-1}(\tilde{M}_h(h)) -\frac{1}{2}T^{n-1}_M(h)\\
                         &=& t_M(h) + \frac{1}{2}t_M(\tilde{M}_h(h))+ \frac{1}{4}T_M^{n-2}(\tilde{M}_h^2(h))- \frac{1}{4}T^{n-2}_M(h)-\frac{1}{2}T^{n-1}_M(h)
\label{Eq:Tak_rec_1}
\end{eqnarray}

\noindent where we have not taken the limit $n \rightarrow \infty$. If we continue to apply equation (\ref{Eq:T_hfull}) we arrive at the recursive definition

\begin{equation}
                T^n_M(h) = \sum_{k=0}^{n} \frac{1}{2^k}t_M\left(\tilde{M}_h^k(h)\right)- \sum_{k=1}^{n}\frac{1}{2^k}T_M^{n-k}(h),
\label{Eq:tak_recurs}
\end{equation}

\noindent In order to simplify equation (\ref{Eq:tak_recurs}) further we define

\begin{equation}
                \tau(n) :=\sum_{k=1}^{n}\frac{1}{2^k}T_M^{n-k}(h).
 \label{Eq:tau(n)}
 \end{equation}

\noindent We can write equation (\ref{Eq:tau(n)}) recursively as

\begin{equation}
                \tau(n)= \frac{1}{2}T_M^{n-1}(h)+\frac{1}{2}\tau(n-1).
\label{Eq:tau(n)recurs}
\end{equation}

\noindent Substituting equations (\ref{Eq:tau(n)}) and (\ref{Eq:tau(n)recurs}) into (\ref{Eq:tak_recurs}) we obtain

\begin{eqnarray}\nonumber
                T^n_M(h) &=& \sum_{k=0}^{n} \frac{1}{2^k}t_M\left(\tilde{M}_h^k(h)\right)- \tau(n).\\
                         &=& \sum_{k=0}^{n} \frac{1}{2^k}t_M\left(\tilde{M}_h^k(h)\right)- \frac{1}{2}T_M^{n-1}(h)-\frac{1}{2}\tau(n-1).
\label{Eq:tak_explicit}
\end{eqnarray}

\noindent Then subbing equation (\ref{Eq:tak_recurs}) back into equation (\ref{Eq:tak_explicit})

\begin{eqnarray}\nonumber
                T^n_M(h) &=& \sum_{k=0}^{n} \frac{1}{2^k}t_M\left(\tilde{M}_h^k(h)\right)- \frac{1}{2}\left(  \sum_{k=0}^{n-1} \frac{1}{2^k}t_M\left(\tilde{M}_h^k(h) \right) - \sum_{k=1}^{n-1}\frac{1}{2^k}T_M^{n-1-k}(h)  \right)-\frac{1}{2}\tau(n-1).\\
                         &=& \sum_{k=0}^{n} \frac{1}{2^k}t_M\left(\tilde{M}_h^k(h)\right)- \frac{1}{2}\left(  \sum_{k=0}^{n-1} \frac{1}{2^k}t_M\left(\tilde{M}_h^k(h) \right)\right)+\frac{1}{2}\tau(n-1) -\frac{1}{2}\tau(n-1).
\label{Eq:tak_explicit2}
\end{eqnarray}

\noindent We then arrive at the expression

\begin{equation}
               T^n_M(h)=\left(\sum_{k=0}^{n} \frac{1}{2^{k+1}}t_M\left(\tilde{M}_h^k(h)\right)\right) + \frac{1}{2^{n+1}}t_M(\tilde{M}_h^n(h))
\label{Eq:tak_recurs_final}
\end{equation}

\noindent which is a useful expression for the Takagi functions as it only involves terms which contain $t_M(x)$. For $0.5 \leq h \leq 1$ we can simplify equation (\ref{Eq:tak_recurs_final}) further by using the fact that $h$ is a fixed point of the modulo $1$ map, $\tilde{M}_h(x)$

\begin{equation}
               T^n_M(h)=\left(\sum_{k=0}^{n} \frac{t_M\left(h\right)}{2^{k+1}}\right) + \frac{1}{2^{n+1}}t_M(h)
\label{Eq:tak_recurs_h0.5}
\end{equation}

\noindent so our expression for the diffusion coefficient in the range $0.5\leq h \leq 1$ is

\begin{eqnarray}\nonumber
                D_M(h)   &=& \frac{h}{2} + \lim_{n\rightarrow \infty}  T_M^{n-1}(h) \\
                         &=& \frac{h}{2} + \lim_{n\rightarrow \infty} \left( \left(\sum_{k=0}^{n-1} \frac{t_M\left(h \right)}{2^{k+1}}\right) + \frac{t_M(h)}{2^{n}} \right).
\label{Eq:D(h)_0.5_1_final_app}
\end{eqnarray}

\end{document}